\documentclass[sigconf, nonacm]{acmart}

\input{datalabmacros} %
\mkclean %

\begin{document}

\title[Is Integer Linear Programming All You Need for Deletion Propagation?]{Is Integer Linear Programming All You Need \\for  Deletion Propagation?}
\titlenote{
Inspired by the 2017 attention paper~\cite{NIPS2017_3f5ee243},
an increasing number of research papers promise that ``X is all you need (for Y).''
Similarly, our conjecture is that \emph{Integer Linear Programs can be designed to solve all PTIME cases of deletion propagation in guaranteed PTIME} and, hence, there is no more need for specialized combinatorial algorithms. 
We give strong evidence of this conjecture by showing that it holds for all currently known tractable cases. 
However, since it is a conjecture, we phrase our title as a question.
}

\newcommand\sbtitle{A Unified and Practical Approach for Generalized Deletion Propagation}
\subtitle{\sbtitle}

\author{Neha Makhija}
\orcid{0000-0003-0221-6836}
\affiliation{%
    \orcidicon{0000-0003-0221-6836}
	Northeastern University\country{USA}}
\email{makhija.n@northeastern.edu}

\author{Wolfgang Gatterbauer}
\orcid{0000-0002-9614-0504}
\affiliation{%
    \orcidicon{0000-0002-9614-0504}
	Northeastern University\country{USA}}
\email{w.gatterbauer@northeastern.edu}

\begin{abstract}

Deletion Propagation (DP) refers to a family of database problems rooted in the classical view-update problem: 
how to propagate intended deletions in a view (query output) 
back to the source database while satisfying constraints and minimizing side effects. 
Although studied for over 40 years, DP variants, their complexities, and practical algorithms have been typically explored in isolation.

This work presents a unified and generalized framework for DP with several key benefits: 
(1) It \emph{unifies and generalizes} all previously known DP variants, 
effectively subsuming them within a broader class of problems, including new, well-motivated variants.
(2)~It comes with a practical and general-purpose algorithm that is ``\emph{coarse-grained instance-optimal}'': 
it runs in PTIME for all known PTIME cases and can \emph{automatically exploit structural regularities} 
in the data,
i.e.\ it does not rely on hints about such regularities as part of the input.
(3) It is \emph{complete}: 
our framework handles all known DP variants in all settings (including those involving self-joins, unions, and bag semantics), 
and allows us to provide new complexity results. 
(4) It is \emph{easy to implement} and, in many cases, outperforms prior variant-specific solutions, sometimes by orders of magnitude. 
We provide the first experimental results for several DP variants previously studied only in theory.

\end{abstract}

\maketitle
\setcounter{page}{1}

\begingroup
\renewcommand\thefootnote{}\footnote{\noindent
This an extended version of the paper with the same title “Is Integer Linear Programming All You Need for Deletion Propagation? A Unified and Practical Approach for Generalized Deletion Propagation” published in Proceedings of the VLDB Endowment, Vol 18, No. 8, (VLDB), 2025.
DOI: \url{https://doi.org/10.14778/3742728.3742756}. 
Our code is available online: \url{https://github.com/northeastern-datalab/generalized-deletion-propagation}. 
This work is licensed under a Creative Commons Attribution International 4.0 License. \\
}\addtocounter{footnote}{-1}\endgroup

\section{Introduction}

Deletion Propagation (DP) was proposed as early as 1982~\cite{Dayal82} and corresponds to a basic view-update problem: 
\emph{Given a tuple we want to delete from a view, which tuples from the source database should we to delete  to accomplish this goal?}
Because multiple view tuples may depend on the same base tuple in the source database, 
deleting that base tuple
can result in unintended side effects beyond the requested deletion.
The challenge is then to delete just enough tuples from the source to achieve the intended view deletion while avoiding unnecessary side effects. This forms a combinatorial optimization problem.
Different optimization goals and different choices about what constitutes a side effect lead to several well-motivated variants of DP that have been studied over the last 40+ years.
Some variants are used for query explainability,
where both \emph{contrastive} or \emph{abductive} explanations~\cite{Marques-Silva2023} can be obtained with 
appropriate choices of side effects and optimization goals.

``Side effects'' are usually measured in the number of tuples affected by a modification. 
Two important types of side effects that have been studied are source side effects~\cite{Buneman:2002,Dayal82} and view side effects~\cite{Buneman:2002,KimelfeldVW12,Kimelfeld12}:
\emph{Source side effects} ($\dpss$) measure the number of tuples deleted from the source database to fulfill the user request, 
while \emph{view side effects} ($\dpvs$) measure the number of unintended tuples deleted from the same view.
A recent variant on the DP-SS problem is the \emph{aggregated deletion propagation} ($\adp$) problem~\cite{ADP} in which a certain number of tuples should be deleted from the view, but it is not specified which tuples.
A different, seemingly unrelated
problem is the recently proposed \emph{smallest witness problem} ($\swp$)~\cite{miao2019explaining,hu_et_al:LIPIcs.ICDT.2024.24}, where a user would like to preserve the view as is, but delete as many tuples from the source as possible.
Although SWP has so far not been understood to be a variant of DP, we show that is problem shares the same structure as other DP problems, can be solved using the same techniques, and -- when combined with other DP problems -- opens up a new space of natural DP variants.

Despite the long history of Deletion Propagation, at least 3 challenges remain. 
This paper shows is that these 3 challenges can be largely addressed by casting the existing problems as special cases of a unified ``General Deletion Propagation'' framework.

\introparagraph{Challenge 1: Countless well-motivated variants}
DP has been studied in many forms over the last 40+ years.
However, one can imagine many more variants that are all well-motivated, and that have not yet been studied. 
These variants can arise from different definitions of side effects, different constraints on allowed side effects, and different optimization goals. 
\cref{ex:intro} gives just one such example of DP that has not been described by prior work
(we explore the wider space of variants more thoroughly in \cref{sec:generalized-deletion-propagation}).

\begin{example}
\label{ex:intro}
An airline company wants to cut costs by reducing the number of flights it offers, and reduce its total operational expenditure by at least $2 \%$.
There are various types of costs incurred by the flight company, such as the fuel cost of the flight and the airport fee at the locations they operate at.
While cutting costs, the airline wants to ensure that it minimizes the effect on its connectivity network i.e.\ pairs of locations that are connected directly or via $1$ layover called ``$(0,1)$-hop connections.''
Additionally, the airline would like to ensure that it maintains a profitable service and so it would like to preserve of all of its most popular connections.

This problem has all the ingredients of a Deletion Propagation problem: the source database is the set of all flights and airports; the view is the set of all location pairs that have direct or 1-hop connection between them. 
The airline would like to delete a certain amount of flight and airport costs (corresponding to cancelling flights, and not having service to an airport) - but it would like to minimize the side effect on the view (the connectivity network) and preserve output tuples of a different view (that shows the most popular connections).
This problem is a mixture of Aggregated Deletion Propagation (which involves deleting an arbitrary fraction of a view), and the Smallest Witness Problem (which involves preserving a view), but is also any extension in many ways (discussed further in \cref{sec:ilp-framework}). For example, the side effects are not measured in the original source or view, but \emph{in a different view (!)}.
\label{ex:introexample}
\label{EX:INTROEXAMPLE}
\end{example}

\introparagraph{Challenge 2: Dissimilar algorithms for similar problems}
Since DP variants have been studied in isolation, the algorithms used to solve these problems are often dissimilar.
Even for one variant, different queries currently require different algorithms.
Thus, new variants are often solved ``from scratch'' and algorithmic insights are not carried over.
DP variants are $\NPC$ (NP-complete) in general, but are $\PTIME$ for certain queries. 
To solve DP for a query optimally, one needs to know the algorithm that can correctly solve the problem variant for the given query in $\PTIME$ (if such an algorithm exists), 
and know that the query and database fulfill the requirements that allow applying the specialized algorithm.
Since algorithms that are specific to the variant and query, they are not generalizable, easily implementable, or extensible to new variants and query classes. 
We are inspired by recent work~\cite{makhija2024unified} that showed that for a particular DP variant called \emph{resilience} (i.e., DP with source side effects for a Boolean query), such a unified framework exists and is guaranteed to terminate in PTIME for all known PTIME cases.
In contrast, we propose a unified ``\emph{coarse-grained instance-optimal}'' framework\footnote{Notice 
that we use \emph{instance-optimal} in a more \emph{coarse-grained} sense than is more common in complexity theory~\cite{Roughgarden_2020}. Our focus is on solving all known PTIME cases in PTIME, but not necessarily using the fastest possible specialized algorithm in each case. In other words, we ignore \emph{fine-grained} complexity that differentiates between different classes within PTIME. We discuss this distinction further in \cref{sec:related-work}.}
which includes \emph{all} previously studied DP variants, including even problems that were not previously phrased as DP ($\swp$), and new variants as well.

\introparagraph{Challenge 3: Algorithms and tractability criteria are unknown for many real-world queries and scenarios} 
DP problems are typically studied for self-join free conjunctive queries under set semantics, because queries with self-joins are known to be notoriously difficult to analyze, and several complexity boundaries have been open for over a decade~\cite{KimelfeldVW12}.
In practice, however, queries often contain unions, are not self-join free, and are executed under bag semantics.
Only very few algorithms and tractability results are known for these more complicated settings,  
such as for queries with self-joins~\cite{DBLP:conf/pods/FreireGIM20,Kimelfeld12},
unions of conjunctive queries~\cite{DBLP:conf/lics/BodirskySL24},
and queries under bags semantics~\cite{makhija2024unified}. 
The overall tractability criterion for queries for such ``real-world'' queries is overall ill-understood.

\introparagraph{Contributions and Outline} 
We solve the challenges outlined above by introducing a unified framework for Deletion Propagation (DP) problems.
We define Generalized Deletion Propagation (GDP), show that this definition encapsulates existing variants as well many natural new variants, 
and give a unified algorithm to solve GDP.
In the process, we recover known tractability results, derive new theoretical results, and provide an experimental validation.

\circled{1}
We define Generalized Deletion Propagation (GDP) in \cref{sec:generalized-deletion-propagation}.
    This definition not only covers all known DP variants, but also 
    includes the Smallest Witness Problem (SWP, which has so far been treated as completely different), and
    covers new well-motivated variants. 
    Our definition allows us to reason about the many DP variants systematically, thus addressing \textbf{Challenge 1}.

\circled{2}
We present an Integer Linear Programming (ILP) formulation for the GDP problem in \cref{sec:ilp-framework}. 
    This formulation allows us to \emph{use one solution for all variants} of DP, thus providing the first step in addressing \textbf{Challenge 2}. 
    The ILP formulation can cover queries with unions and self-joins, and both the set and bag semantics settings, thus giving a valuable tool to address \textbf{Challenge 3}.

\circled{3}
While providing ILP formulations is a typical approach for solving NPC optimization problems, our key technical contribution addressing \textbf{Challenge 3} is proposing an ILP \emph{with the right algorithmic properties}: 
    We show in \cref{sec:tractibility-results} that for \emph{all known PTIME cases}, our ILP formulation is solvable in PTIME via an LP relaxation. 
    Thus, we do not need dedicated PTIME algorithms for special cases; our theory shows that standard ILP solvers default to solving these cases in PTIME.
    This means that the ILP framework can be directly used to solve all tractable instances of DP, thus resolving \textbf{Challenge 2} for all known PTIME cases.
    Notice that it is \emph{not trivial} to come up with the right ILP formulation. We show that a more obvious ILP formation does not have the desired PTIME guarantees, and can be over $2$ orders of magnitude slower in practice.

\circled{4}
    We uncover a new tractable case for well-known variants of the DP problem, thus showing that our framework is a powerful tool to address \textbf{Challenge 3}, the long-standing challenge of capturing the exact tractability boundary.
    Concretely, we prove in 
    \cref{SEC:TRACTABILITYNEW}
    that the ILP formulation 
    of a query with union and self-join that 
    can be solved in PTIME under bag semantics.
    
\circled{5}    
    We experimentally evaluate the efficiency of our contributions in \cref{sec:expts}. Our approach performs comparably and sometimes even better than specialized algorithms for particular DP variants, and can solve new tractable cases that were not previously known.

We only provide a proof intuition for each theorem in the main text, and make full proofs, additional details and experiments available in the appendix.

\section{Preliminaries}
\label{sec:prelims}

\introparagraph{Standard database notations}
A \emph{conjunctive query} (CQ) is a first-order formula $Q(\vec y)$ $= \exists
\vec x\,(g_1 \wedge \ldots \wedge g_m)$ 
where the variables $\vec x = \langle x_1, \ldots, x_\ell \rangle$ 
are called existential variables,
$\vec y$ are called the head or free variables,
and each atom $g_i$ represents a relation 
$g_i= R_{j_i}(\vec x_i)$ where $\vec x_i \subseteq \vec x \cup \vec y \cup U$, with $U$ being a universe of constant values. 
$\var(X)$ denotes the variables in a given relation/atom.
Notice that a query has at least one output tuple iff the Boolean variant of the query 
(obtained by making all the free variables existential) is true.
A \emph{self-join-free CQ} (SJ-free CQ) is one where no relation symbol occurs more than once and thus every atom represents a different relation. 
A \emph{union over conjunctive queries} (UCQ) is given by $Q(\vec y) \datarule \bigcup_{i \in [1, l]} \exists \vec x_i (g^i_1 \wedge \ldots \wedge g^i_m)$ where for each $i$, $Q(\vec y) \datarule \vec x_i (g^i_1 \wedge \ldots \wedge g^i_m)$ is a CQ.
We write $\mathcal{D}$ for the database, i.e.\ the set of tuples in the relations.
When we refer to bag semantics, we allow $\D$ to be a multiset of tuples in the relations.
Unless otherwise stated, a query in this paper refers to a UCQ, and a database instance $\mathcal{D}$ can be considered to a multiset. 
However, we may fudge notation and represent $\D$ as a set of tuples if all the multiplicities are 1.

We write $[\vec w / \vec x]$ as a valuation (or substitution) of query variables $\vec x$ by $\vec w$.
A view tuple or an output tuple $v$ is a valuation of the head variables $\vec y$ that is permitted by $\D$.
Similarly, a \emph{witness} $w$ is a valuation of all variables $\vec x$ that is permitted by $\mathcal{D}$\footnote{Note that this differs from an alternate notion of witness~\cite{Buneman:2002,hu_et_al:LIPIcs.ICDT.2024.24} which defines the witness of an output tuple as a subset of input tuples such that running a query over the subset produces the output tuple}.
We can alternately describe a witness as an output tuple for the \emph{full version} of the query $Q$, which is obtained by making all the free variables existential.
We denote the set of views tuples obtained by evaluating a query $Q$ over a database $\D$ simply as $Q(\D)$.
For example, consider the 2-chain query 
$Q^\infty_2(x) \datarule R(x, y), S(y, z)$
over the database 
$\mathcal{D} = \{r_{12}{:\,}R(1,2), r_{2,2}{:\,}R(2,2), s_{2,3}{:\,}S(2,3)\}$.
Then $Q(\D) = \{Q(1), Q(2)\}$ and
 $\witnesses(Q^\infty_2, D) =$ $\{(1, 2, 3), (2, 2, 3)\}$.

\introparagraph{Linear Programs (LP)}
Linear Programs are standard optimization problems \cite{aardal2005handbooks, schrijver1998theory}
in which the objective function and the constraints are linear.
A standard form of an LP is $\min \vec c^\transpose \vec x$ s.t. $\vec W \vec x \geq \vec b$, where $\vec x$ 
denotes the variables, the vector $\vec c^\transpose$ denotes weights of the variables in the objective, the matrix $\vec W$ denotes the weights of $\vec x$ for each constraint, and $\vec b$ denotes the right-hand side of each constraint.
The objective function $f=\vec c^\transpose \vec x$ may also be referred to as a \emph{soft constraint}. 
We use $f^*$ to denote the optimal value of the objective function.
If the variables are constrained to be integers, the resulting program is called an Integer Linear Program (ILP).
The \emph{LP relaxation} of an ILP program is obtained by removing the integrality constraint for all variables.

\section{Related Work}
\label{sec:related-work}

We will discuss in \cref{sec:generalized-deletion-propagation} related work on problems that fit within the umbrella of DP. 
This section covers additional related work concerning broader themes that are discussed in this paper.

\introparagraph{Reverse Data Management (RDM)}
DP can be seen as a type of reverse data management problem~\cite{DBLP:journals/pvldb/MeliouGS11}. 
RDM problems search for optimal interventions in the input data that would lead to a desired output.
RDM problems are useful in many applications, such as intervention-based approaches for explanations~\cite{DBLP:journals/pvldb/HerschelHT09,DBLP:journals/pvldb/0002M13,SudeepaSuciu14,glavic2021trends}, fairness~\cite{salimi2019interventional,galhotra2017fairness}, 
causal inference~\cite{galhotra2022hyper}, and data repair~\cite{wang2017qfix}.
The Tiresias system~\cite{Meliou2012_tiresias} solves how-to problems, a type of reverse data management problem using Mixed Integer Linear Programming (MILP). 
However, its focus is on building the semantics of a query language for how-to problems that can be translated to an MIP, and not on building a unified method that can recover tractable cases.

\introparagraph{Intervention-Based Explanations} 
Formal Explainability in AI (FXAI)~\cite{Marques-Silva2023} 
distinguishes between two types of explanations:
\emph{Abductive} explanations (or locally sufficient reasons \cite{pmlr-v235-bassan24a}) identify a minimal subset of features that, when fixed to their original values, are sufficient to \emph{guarantee the original prediction}.
They are also known as `Why?' explanations as they explain why a prediction is the way it is.
\emph{Contrastive} explanations 
identify a minimal subset of features that, when altered from their original values, are sufficient to \emph{change the original prediction}.
They are also known as `Why not?' explanations 
as they explain why the prediction is not different from what it is.
These notions also extend to relational query explanations,
and we can interpret the Smallest Witness Problem (SWP)~\cite{miao2019explaining,hu_et_al:LIPIcs.ICDT.2024.24} 
as an instance of abductive explanation, 
and the Resilience Problem (RES)~\cite{makhija2024unified} as a contrastive explanation. 
Generalized Deletion Propagation (GDP) subsumes both SWP and RES and can give \emph{both abductive and contrastive explanations}
in the same framework.
Notice that `Why' and `Why not' explanations have been understood differently in the context of database provenance~\cite{DBLP:journals/pvldb/MeliouGMS11,DBLP:journals/corr/abs-0912-5340}: 
`Why' has been used to understand why a given tuple is in the output (a `prediction' is true)
whereas `why not' to understand why a tuple is not in the output (a `prediction' is false).

\introparagraph{Linear Optimization Solvers}
A key practical advantage of modeling problems as ILPs is that
there are many highly-optimized ILP solvers, both commercial~\cite{gurobi} and free~\cite{mitchell2011pulp} which can obtain exact results efficiently, in practice.
ILP formulations are standardized, and thus programs can easily be swapped between solvers.
Any advances made over time by these solvers can automatically make implementations of these problems better over time.
For our experimental evaluation we use Gurobi\footnote{Gurobi offers a free academic license \url{https://www.gurobi.com/academia/academic-program-and-licenses/}.} which uses an LP based branch-and-bound method to solve ILPs~\cite{gurobi_working}. 
This means that it first computes an LP relaxation bound and then explores the search space to find integral solutions that move closer to this bound.
If an integral solution is encountered whose objective is equal to the LP relaxation optimum,
then the solver has found a guaranteed optimal solution and is done.
In other words, if we can \emph{prove that the LP relaxation of our given ILP formulation has an integral optimal solution}, then we are guaranteed that our original ILP formulation will terminate in $\PTIME$, even without changing the ILP formulation or letting the solver know about the theoretical complexity.

\introparagraph{Complexity of solving ILPs}
Solving ILPs is NPC~\cite{karp1972reducibility}, while LPs can be solved in $\PTIME$ with Interior Point methods \cite{grotschel1993ellipsoid,cohen2021solving}.
The specific conditions under which ILPs become tractable is an entire field of study.
It is known that if there is an optimal integral assignment to the LP relaxation, then the original ILP can be solved in $\PTIME$ as well.
There are many structural characteristics that define when the LP is guaranteed to have an integral minimum, and thus where ILPs are in $\PTIME$. 
For example, if the constraint matrix of an ILP is \emph{Totally Unimodular} \cite{schrijver1998theory} then the LP always has the same optima. 
Similarly, if the constraint matrix is \emph{Balanced}~\cite{conforti2006balanced}, several classes of ILPs are $\PTIME$.
We do not use any of these techniques in this paper, but we believe future research in this area may help \emph{automatically} identify more tractable cases of DP problems.

\introparagraph{ILPs and Constraint Optimization in Databases} 
Integer Linear Programming has been used in databases for problems such as in solving package queries~\cite{brucato2019scalable}, 
query optimization~\cite{DBLP:conf/sigmod/Trummer017},
and general optimization applications~\cite{DBLP:conf/ssdbm/SiksnysP16}.
However, other than our recent work on the resilience problem~\cite{makhija2024unified}, we are unaware of any work in databases that uses ILPs to automatically recover tractable cases by proving that the condition ILP = LP holds for the PTIME cases, i.e.\ that the LP relaxation has an optimal integral value and thus the original ILP problem can be solved in guaranteed PTIME.
We show that a straightforward application of that earlier idea to our generalized problem formulation 
does not work as the LP relaxation of the naive formulation can give fractional optimal solutions
(see \cref{ex:smoothing,Fig:smoothing:illustration}).
In \cref{sec:wildcard,sec:smoothing} 
we develop
\emph{new techniques that allowed us to prove that the natural LP relaxation of the resulting non-obvious ILP formulation has the ILP = LP property}.
We also show the effect in our experiments \cref{fig:expt-3}
with a reduction from over 3 hours to under 20 seconds.

\introparagraph{Instance Optimal Algorithms} 
Our notion of ``\emph{coarse-grained instance optimality}'' is inspired by the notion of instance optimality in complexity theory~\cite{Roughgarden_2020}.
The need for instance optimality or beyond worst case complexity analysis has been increasingly recognized since worst-case complexity analysis can be overly pessimistic and fails to capture the efficient real world performance of many algorithms such as in ILP optimization and machine learning.
Instance optimal algorithms have also been sought for some problems in databases such as top-$k$ score aggregation~\cite{Fagin2001Optimal}, and join computation~\cite{Khamis2016Joins,Ngo2014Beyond,DBLP:conf/icdt/AlwayBS21}.

\section{Generalized Deletion Propagation}
\label{sec:generalized-deletion-propagation}

\begin{figure*}[t]
  \centering
  \includegraphics[scale=0.46]{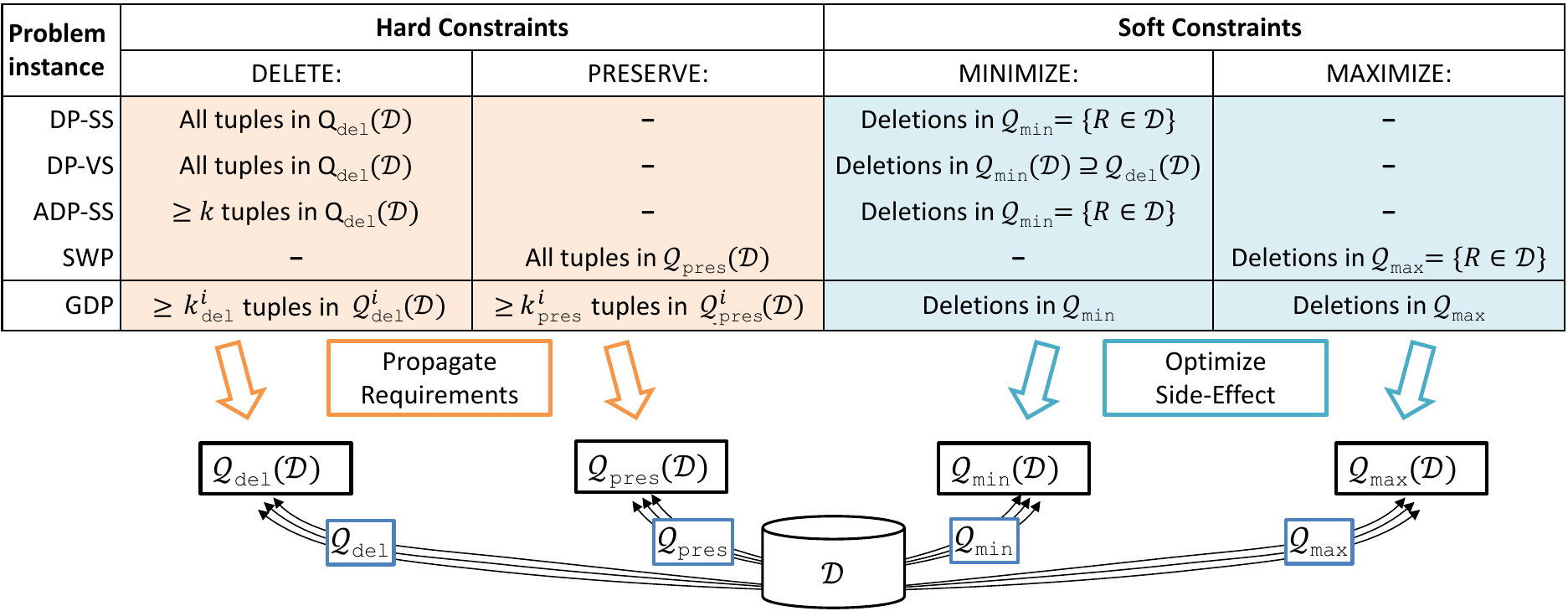}
  \caption{
  Generalized Deletion Propagation (GDP) is defined over 4 different sets of views, two of which model hard constraints, and the other two model soft constraints (optimization objectives).
  Our approach encapsulates previously studied NPC variants of the deletion propagation problem as special cases: Deletion Propagation with Source Side Effects ($\dpss$) \cite{Dayal82}, Deletion Propagation with View Side Effects ($\dpvs$) \cite{Kimelfeld:2013}, Aggregated Deletion Propagation with Source Side Effects ($\adp$) \cite{ADP}, and Smallest Witness Problem ($\swp$) \cite{hu_et_al:LIPIcs.ICDT.2024.24}. 
  Notice that GDP is a generalization of the prior variants in multiple senses: 
  1) It allows for side effects on a view different from the original.
  2) It allows each type of constraint to be enforced over multiple views.
  3) It allows for a combination of constraints and measured side effects.
  }
  \label{fig:deletion-propagation}
\end{figure*}

We introduce Generalized Deletion Propagation (GDP) which generalizes all prior variants of deletion propagation, and also allows for new variants to be defined.
The new variants are motivated by the following observations:
(1) The number of deletions in the source or view are not the only possible side effects; 
one could care about side effects on \emph{another view} that is different from the one in which the deletion occurs. 
(2) It is natural to enforce constraints or optimize side effects over multiple views.
(3) Prior variants focus on a specific type of constraint (either deletion or preservation). In practice, one might want to combine these constraints (e.g., minimizing deletions from one view while maximizing deletions from another).
These extensions are motivated with examples in \cref{sec:gdp-new-variants}.

We observe that with 4 different sets of views, we can model all existing problems and can also combine individual constraints in arbitrary ways. 
\Cref{def:deletion-propagation} thus defines \emph{generalized deletion propagation} as a constraint optimization problem over four set of views.
These sets of views correspond to four primitive operations (or requirements) that typically occur in deletion propagation variants - a requirement to delete tuples from a view, preserve tuples in a view, minimize side effects on a view, or to maximize side effects on a view.
\cref{sec:gdp-capture-prior-variants} discusses how the GDP definition encapsulates all past variants of DP as special cases (also depicted in \cref{fig:deletion-propagation}), while \cref{sec:gdp-new-variants} motivates DP new variants that are captured by GDP.

\subsection{Defining Generalized Deletion Propagation}

Before we define GDP, we introduce some notation.
We use bold notation for vectors (as in $\vec x$) and superscript for entries (as in $x^i$).
$\Q$~represents an ordered set of queries, and $Q^i$ represents the $i^{\textrm{th}}$ query in $\Q$.
$|Q(\D)|$ is defined as the number of output tuples in $Q(\D)$ and 
$|\Q(\D)|=\sum_{Q \in \Q} |Q(\D)|$ as the number of output tuples across all views in $\Q$.
We define $\Delta Q(\D,\Gamma)$ as the set of output tuples in $Q(\D)$ that are deleted as a consequence of deleting $\Gamma$ from the database $\D$ and hence are not present in $Q(\D \setminus \Gamma)$. 
Similarly, we define $\Delta \Q(\D,\Gamma)$ as the set of tuples deleted from all views in $\Q$:
\begin{equation*}
  | \Delta\Q (\D,\Gamma) | = \sum_{Q^i \in \Q} |Q^i(\D)|-|Q^i(\D \setminus \Gamma)|
\end{equation*}

  \begin{definition}[Generalized Deletion Propagation (GDP)]
    Given four ordered sets of monotone queries $\querydel, \querypres, \querymin$ and $\querymax$ over a database $\D$, 
    and vectors of positive integers $\kdel$ and $\kpres$ of size 
    equal to the number of views in $\querydel$ and $\querypres$ respectively,
    the GDP problem is the task of determining a set of input tuples $\Gamma \subseteq \D$ such that 
    \begin{equation*}
      |\Delta\querymin(\D, \Gamma)| -  |\Delta\querymax(\D, \Gamma)|
    \end{equation*}
    is minimized and the following hard constraints are satisfied:

    \begin{enumerate}[leftmargin=*]
    
    \item 
    Deleting $\Gamma$ from the database $\D$ deletes at least $\kdelparam{i}$ 
    output tuples from the $i^{\textrm{th}}$ view defined by $\querydel$ i.e., 
    \begin{equation*}
      |\querydelparam{i}(\D \setminus \Gamma)| \leq |\querydelparam{i}(\D)| - \kdelparam{i}
    \end{equation*}
    
    \item 
    Deleting $\Gamma$ from the database $\D$ preserves at least $\kpresparam{i}$ output tuples from the $i^{\textrm{th}}$ view defined by $\querypres$ i.e., 
    \begin{equation*}
      |\querypresparam{i}(\D \setminus \Gamma)| \geq \kpresparam{i} \hspace{14.5mm}
    \end{equation*}
  \end{enumerate}
    \label{def:deletion-propagation}
  \end{definition}

\subsection{Capturing Prior Variants of Deletion Propagation with GDP}
\label{sec:gdp-capture-prior-variants}

We next show how each of the previously studied variants of the deletion propagation problem is a special case of GDP.

\subsubsection{Deletion Propagation with Source Side Effects ($\dpss$)~\cite{Buneman:2002,Dayal82} and Resilience (RES) \cite{DBLP:conf/pods/FreireGIM20,DBLP:journals/pvldb/FreireGIM15,DBLP:conf/lics/BodirskySL24,makhija2024unified,MIAO2020_resilience,MIAO2020_resilience}}
\label{sec:gdp-dp-ss}
Deletion Propagation with source side effects 
($\dpss$) is one of the two originally formulated variants of the deletion propagation problem~\cite{Buneman:2002}.

\begin{definition}[$\dpss$]
\label{def:dpss}  
  Given a view defined by a query $Q$ over a database $\D$, and an output tuple $t \in Q(\D)$, 
  the deletion propagation with source side effects problem is the task of determining a set of input tuples $\Gamma \subseteq \D$ such that $|\Gamma|$ is minimized and $t$ is not contained in $Q(\D \setminus \Gamma)$. 
  In other words,
  \begin{equation*}
    \min |\Gamma| \textrm{ s.t.\ } t \notin Q(D \setminus \Gamma)
  \end{equation*}
\end{definition}

$\dpss$ is a special case of GDP - we can solve a $\dpss$ problem by setting $\querydelparam{1}$ to be a query with constants selecting for the values of $t$, $\kdelparam{1} = 1$ and setting $\querymin$ to be the set of identity queries that select all tuples from any relation in $\D$.
The key observation is that source side effects can also be represented by computing a set of queries $\querymin$, and then the difference between source and view side effects results from the choice of query that defines the view.

Resilience (RES) is a variant of $\dpss$ that focuses on Boolean queries and asks for the minimum number of deletions needed to make a query false. 
It has been called the ``simplest'' of all deletion propagation problems \cite{DBLP:journals/pvldb/FreireGIM15}, and a large amount of literature has been dedicated to understanding its complexity \cite{DBLP:journals/pvldb/FreireGIM15,DBLP:conf/pods/FreireGIM20,MIAO2020_resilience,DBLP:conf/lics/BodirskySL24,makhija2024unified}.
The complexity results for the RES problem also imply complexity results for the $\dpss$ problem~\cite{DBLP:journals/pvldb/FreireGIM15}.
Existing work has shown a complexity dichotomy for self-join free conjunctive queries, 
both under set \cite{DBLP:journals/pvldb/FreireGIM15} and bag semantics \cite{makhija2024unified}, 
yet only few tractability results for queries with self-joins and unions are known~\cite{DBLP:conf/pods/FreireGIM20,makhija2024unified,DBLP:conf/lics/BodirskySL24}.
The RES problem can be modelled as a special case of GDP similarly as $\dpss$, with the added restriction that 
$\querydelparam{1}$ 
is a boolean query.

\subsubsection{Deletion Propagation: View Side Effect ($\dpvs$)
\cite{Buneman:2002,Kimelfeld:2013,Kimelfeld12,KimelfeldVW12}}
Deletion Propagation with View Side effects ($\dpvs$) has the same deletion propagation requirement (or ``hard constraint'') as $\dpss$, 
but does so with the goal of minimizing the side effects on the view in which the deletion occurs.

\begin{definition}[$\dpvs$]
  \label{def:dpvs}  
  Given a view defined by a query $Q$ over a database $\D$, and an output tuple $t \in Q(\D)$, 
  the deletion propagation with view side effects problem 
  is the task of determining a set of input tuples $\Gamma \subseteq \D$ such that $|\Delta Q(\D,\Gamma)|$ is minimized and $t$ is not contained in $Q(D \setminus \Gamma)$. 
  In other words,
  \begin{equation*}
    \min |Q(\D)| - |Q(\D\setminus\Gamma)| 
    \textrm{ s.t.\ }  t \notin Q(\D \setminus \Gamma)
  \end{equation*}
\end{definition}

$\dpvs$ is a special case of GDP where $\querydel(\D)$ contains a single query whose output is the output tuple $t$ (just like in $\dpss$), 
$\kdelparam{1} = 1$, 
and $\querymin$ has as single query $Q$ from the original $\dpvs$ problem.
A complexity dichotomy has been shown for the $\dpvs$ problem for self-join free CQs under set semantics \cite{KimelfeldVW12}.

\subsubsection{Aggregated Deletion Propagation with Source Side effect ($\adp$) \cite{ADP}} 
The Aggregated Deletion Propagation ($\adp$) formulation extends the previous $\dpss$ by requiring the deletion of any $k$ output tuples from a view, rather than a specific output tuple.

\begin{definition}[ADP]
  \label{def:adp}  
  Given a view defined by a query $Q$ over a database $\D$, 
  and a positive integer $k$, the Aggregated Deletion Propagation (ADP) problem 
  is the task of determining a set of input tuples $\Gamma \subseteq \D$
  such that $|\Gamma|$ is minimized and at least $k$ tuples are removed from $Q(\D)$ as a consequence of removing $\Gamma$ from $\D$. 
  In other words,
  \begin{equation*}
    \min |\Gamma| \textrm{ s.t.\ }
    Q(\D\setminus \Gamma)| \leq |Q(\D)| - k
  \end{equation*}
  \end{definition}

Even though not explicit in the name of the problem, $\adp$ cares about minimizing source side effects (which can be captured by GDP in the same manner as for $\dpss$).
A complexity dichotomy has been shown for the ADP problem for self-join free conjunctive queries under set semantics \cite{ADP}.

\subsubsection{Smallest Witness Problem (SWP) \cite{miao2019explaining}}
The Smallest Witness Problem was not proposed as a DP problem, but was noted to bear a strong but unspecified resemblance to the deletion propagation variants \cite{hu_et_al:LIPIcs.ICDT.2024.24}. 
We show that this resemblance is due to the fact that -- when modelled as a constraint optimization problem -- the correspondence of deletions of input and output tuples are based on exactly the same constraints.
Concretely, SWP can be seen as a ``preservation propagation'' problem, where the goal is to find the largest set of tuples that can be removed from the database without affecting the results of a query.
Using negation, we reformulate this as minimization problem (to maintain consistency with other definitions in this section):

\begin{definition}[SWP]
\label{def:swp}  
  Given a view defined by a query $Q$ over a database $\D$, the smallest witness problem is the task of determining a set of input tuples $\Gamma \subseteq D$ such that $|\Gamma|$ is maximized and $\Delta V(\D,\Gamma)$ is exactly $0$. 
  In other words,
  \begin{equation*}
    \min -|\Gamma|
    \textrm{ s.t. }
    |Q(\D\setminus \Gamma)| = |Q(\D)|
  \end{equation*}
\end{definition}

A complexity dichotomy has been shown for the SWP problem for self-join free conjunctive queries under set semantics \cite{hu_et_al:LIPIcs.ICDT.2024.24}. 
Interestingly, the tractable cases for SWP are a subset of the tractable cases for $\dpvs$, reaffirming that these variants have a structural connection and should be studied together.

\subsection{Capturing Natural New Variants of Deletion Propagation with GDP}
\label{sec:gdp-new-variants}

Our GDP formulation allows for the definition of new variants of the deletion propagation problem based on at least three types of extensions. 
These extensions can be combined in arbitrary ways, leading to a rich set of new problems.

\introparagraph{Extension 1: New types of side effects} Existing DP variants focus on minimizing source side effects or view side effects. 
However, one can easily imagine a user wanting to delete tuples from one view while minimizing side effects on another view. 
For instance, in \cref{ex:introexample}, where the output view that deletion constraints very defined on (the connectivity network), was different from the view side effects were measured on (view of popular connections).

\introparagraph{Extension 2: Constraints over multiple views}
Existing DP variants focus on a single view from which deletions are propagated. 
However, one can imagine a scenario where tuples from multiple views are to be deleted. 
As we saw in \cref{ex:introexample}, the airline wanted to cut down on multiple costs such as fuel costs and airport lease costs.
Depending on the current structure of the airline, a different percent of cost-cutting in each category may be required, and it is always better to jointly optimize over all the expense views\footnote{
Note that performing deletions on multiple views one at a time is not the same as performing deletions on all views simultaneously, and the side effects of performing DP on each view independently may be higher than performing DP on all views simultaneously.
 Cutting $2\%$ of total costs is not necessarily the same as cutting $1\%$ of fuel costs and $1\%$ of airport lease costs. 
}.

\introparagraph{Extension 3: Combination of Deletion and Preservation Constraints}
Current DP variants focus on either deletion constraints ($\dpss$, $\dpvs$, $\adp$) or preservation constraints ($\swp$) exclusively.
However, one may want to enforce both deletion and preservation constraints simultaneously - like in \cref{ex:introexample} where it matters to cut down on costs but also preserve the popular routes.

\section{ILP Framework for GDP}
\label{SEC:ILP}
\label{sec:ilp-framework}

This section specifies an Integer Linear Program (ILP) $\gdpilp$ which returns an optimal solution for $\gdp$ for any instance supported by \cref{def:deletion-propagation}. 
We proceed in three steps, first providing a basic ILP formulation and subsequently improving it in two steps.
Our approach works even if some views are defined \emph{with self-joins}, or if the underlying database uses \emph{bag semantics}. 
We focus in this section on proving correctness, while
\cref{sec:tractibility-results} later investigates how known tractable cases can be solved in PTIME, despite the problem being NPC in general.
The input to the $\gdpilp$ are the four sets of view-defining queries $\querydel$, $\querypres$, $\querymin$, $\querymax$
over a database $\mathcal{D}$.
Note that any of these sets can be empty as well.\footnote{Notice that 
the problem is still defined (though trivial) even if all sets are empty: 
Then any set of interventions satisfy the problem, 
and the objective value is always $0$.}
As input to our computation, we also assume as given the \emph{set of witnesses} for each output tuple in any of the computed views, which
can be obtained in PTIME by running the \emph{full version} $Q^F$
of each 
query $Q$ and computing the associated provenance polynomial. 
The full version $Q^F$ of a query $Q$ is the query that we get by 
making any existential variables into head variables (or equivalently, by removing all projections). 
For example, 
the full version of
$Q(x) \datarule R(x,y), S(y,z)$
is $Q^{F}(x,y,z) \datarule R(x,y), S(y,z)$.
The use of witnesses as an intermediary between input (database) and output (view) tuples is a key modeling step that allows us to formulate DP problems with linear constraints.
We thus associate with each output tuple a set of witnesses and use these sets of witnesses to construct $\gdpilp$.

In a slight abuse of notation we write $v \in \mathcal{Q}(\D)$ for $v \in \bigcup_{Q \in \mathcal{Q}}$ ${Q(\D)}$ and similarly, $w \in \mathcal{Q}^F(\D)$ for $w \in \bigcup_{Q \in \mathcal{Q}}{Q^F(\D)}$.
We write that $v \in w$ if $v \in Q(\D)$ is a projection of $w \in Q^F(\D)$ onto the head variables of $Q$.
For example, for the earlier example of 
$Q(x) \datarule R(x,y), S(y,z)$
and $Q^{F}(x,y,z) \datarule R(x,y), S(y,z)$,
assume we have 
two witnesses
$w_1 = Q^F(1,2,3)$,
$w_2 = Q^F(1,3,2)$,
$w_3 = Q^F(2,1,3)$,
and two view tuples
$v_1 = Q^1(1)$,
$v_2 = Q^1(2)$.
Then 
$v_1 \in w_1$,
$v_1 \in w_2$,
$v_1 \not \in w_3$,
$v_2 \not \in w_1$, etc.
It is very important to note that we treat output tuples of different views as distinct, even if they correspond to the same set of tuples in the input database.
Thus, we can have 
$v_1 \in \querydelparam{i}(\D)$ 
and $v_2 \in \querypresparam{j}(\D)$ 
with $\querydelparam{i} = \querypresparam{j}$, 
and the valuation of variables for $v_1$ is the same as for $v_2$, 
but we will still treat them as distinct: $v_1 \neq v_2$
(similarly for views).
Notice that this modeling decision appears at first to create inconsistencies, as our algorithm theorectically permits $v_1$ to be deleted from the view
while $v_2$ is preserved. 
However, as we discuss later, this \emph{does not create inconsistencies} and is actually \emph{crucial} for the tractability proofs in \cref{sec:tractibility-results}.

\subsection{A basic ILP Formulation for $\gdp$}

We first define a naive ILP $\gdpilpbasic$ with three components:
the ILP variables, an ILP objective function, and ILP constraints.

\subsubsection{ILP Variables} 
\label{sec:ilpvariables}
We introduce binary variables $X[t]$ for each input tuple $t$ in the relations from $\D$
which takes on value $1$ if the corresponding tuple is deleted, and $0$ otherwise.
Similarly, we introduce binary variables 
$X[v]$ for each output tuple $v$ in each of the view-defining queries in 
$\querydel(\D), \querypres(\D), \querymin(\D), \querymax(\D)$, 
and $X[w]$ for each witness $w$ in the full version of those queries.

\subsubsection{ILP Objective Function (``Soft constraints'')} 
The only possible side effects of deleting a set of input tuples 
on a view defined by a monotone query are deletions of tuples in the view.
As defined in \cref{def:deletion-propagation}, we thus count the side effects as the number of output tuples deleted from $\querymin(\D)$ plus the number of tuples preserved in $\querymax(\D)$, respectively.
Minimizing the number of tuples preserved in a view 
is equivalent to maximizing the number of tuples deleted in that view,
which is equivalent to minimizing $-1$ times the number of tuples deleted in that view.
Thus, our overall goal is to minimize the following objective function:
\begin{equation*}
    f(\vec X) =     \sum_{v \in \querymin(\D)}  \!\!\!X[v]\,\,\, - \!\!\! \sum_{v \in \querymax(\D)} \!\!\!X[v]
\end{equation*}

\subsubsection{ILP Constraints (``Hard constraints'')}
The basic ILP formulation has two types of constraints: 
(1) \emph{User constraints (UCs)} are those that are application-specific and are specified by the user.
(2)~\emph{Propagation constraints (PCs)} encode the various relationships between tuple variables, witness variables, and view variables needed for consistency.
In other words, PCs capture the effect of the hard user constraints on the input database, and then the effect of the input database on various views.

\smallsection{(1) User constraints (UCs)} 
These are the deletion and preservation constraints that are specified by the user on the view definitions $\querydel$ and $\querypres$, respectively.
    The deletion constraints specify that at least $\kdelparam{i}$ tuples must be deleted from each view $\querydelparam{i} \in \querydel$, 
    while the preservation constraints specify that at least $\kpresparam{i}$ tuples must be preserved in each view $\querypresparam{i} \in \querypres$ (which is equivalent to deleting at most $|\querypresparam{i}(\D)| - \kpresparam{i}$ tuples):
    \begin{align*}
        \sum_{v \in \querydelparam{i}(\D)} \!\!\! X[v]  
            &\geq \kdelparam{i}                    
            &&\forall \querydelparam{i} \in \querydel    \\
        \sum_{v \in \querypresparam{i}(\D)} \!\!\! X[v]    
            &\leq |\querypresparam{i}(\D)| - \kpresparam{i}    
            &&\forall \querypresparam{i} \in \querypres
    \end{align*}

    \begin{figure}
        \centering
        \includegraphics[scale=0.26]{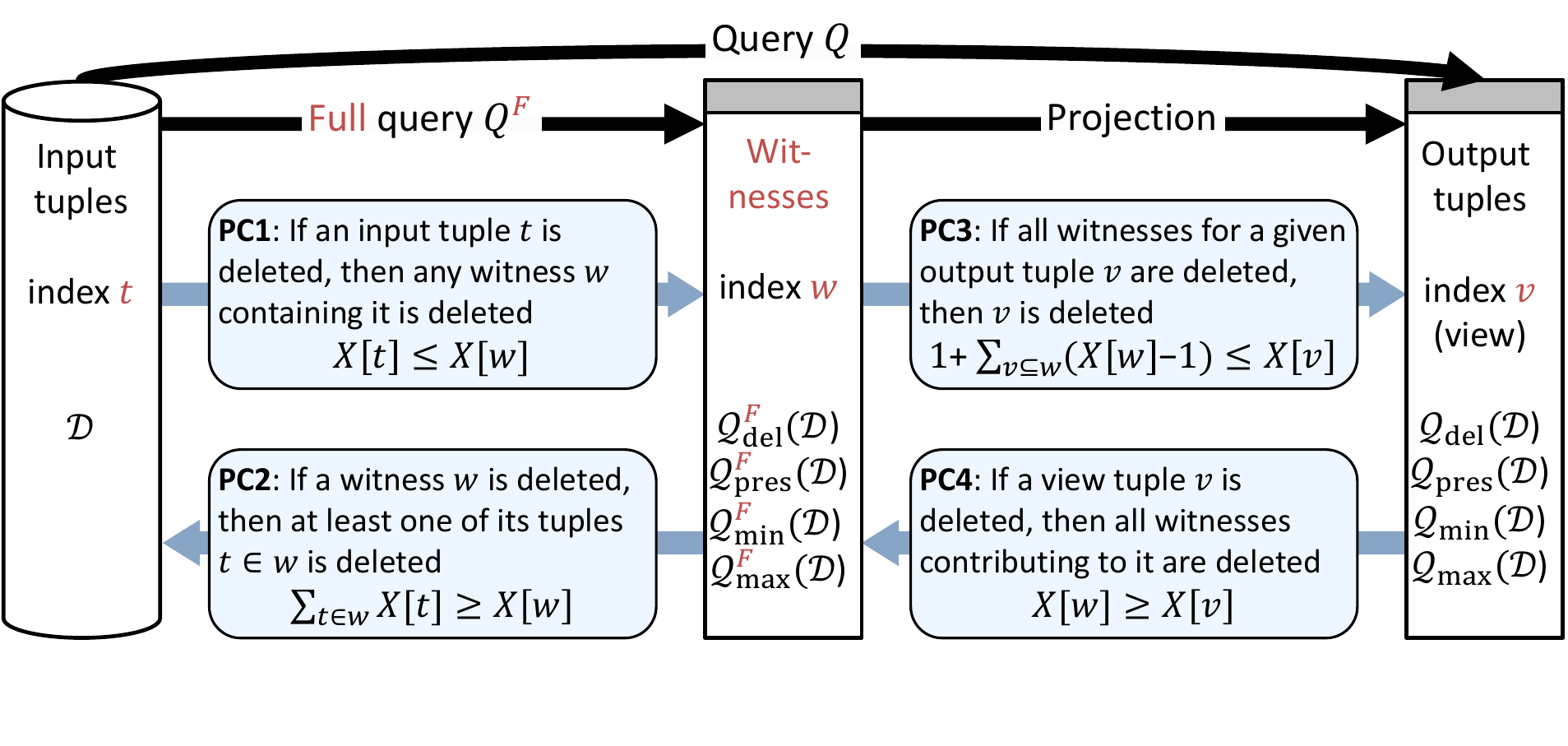}
        \vspace{-3mm}
        \caption{
        Propagation constraints in our ILP formulation $\gdpilpbasic$, explained in the direction of propagating deletions and thus providing lower bounds on the variables.
        The witness variables are the bridge between the tuple variables and the view variables, and represent the output tuples of the corresponding full query.}
        \label{fig:ilp-constraints}
    \end{figure}

\smallsection{(2) Propagation constraints (PCs)}
These constraints encode the relationships between input tuples, witnesses, and tuples in the views to obtain upper and lower bounds on each. 
Any deletion in a view needs to be reflected also in the input database, and as consequence also in the other views. It is this necessary ``\emph{propagation of deletions}'' 
from views (output tuples) to the database (input tuples) that 
gave this family of problems its name~\cite{Buneman:2002}.

\Cref{fig:ilp-constraints} shows a summary of the propagation constraints, split into two parts: the propagation constraints between 
input tuple variables and witness variables (PC1 and PC2), 
and between witness variables and view variables (PC3 and PC4).
Notice that all PCs are bidirectional in that they compare two types of variables and give an upper bound for one and a lower bound for the other. 
Thus, each constraint can be explained in two ways (depending on the direction of the propagation),
but not all constraints need to be applied to all views (recall our wildcard semantics). 
We first describe the constraints, and then discuss when they are enforced.

    \begin{itemize}[leftmargin=1pc]
        \item \textbf{PC1:} 
        ($\rightarrow$)
        If an input tuple $t$ is deleted, then any witness $w$ containing it is deleted.
        ($\leftarrow$) 
        If a witness $w$ is not deleted, then neither of its tuples $t \in w$ is deleted.
        \begin{equation*}
            X[t] \leq X[w], t \in w
        \end{equation*}

        \item \textbf{PC2:} 
        ($\leftarrow$) 
        If a witness $w$ is deleted, then at least one of its tuples $t \in w$ is deleted.
        ($\rightarrow$) 
        Alternatively, if all tuples $t \in w$ are not deleted, then the witness $w$ is not deleted.
        \begin{equation}
            \sum_{t \in w} X[t] \geq X[w]
        \end{equation}

        \item \textbf{PC3:} 
        ($\rightarrow$)
        If all witnesses for a given output tuple $v$ are deleted, then $v$ is deleted.
        ($\leftarrow$) 
        If $v$ is not deleted,
        then at least one witness $w$ for $v$ is not deleted.
        \begin{equation*}
            1+\sum_{v \subseteq w} \big(X[w] - 1\big) \leq X[v]
        \end{equation*}

        \item \textbf{PC4:} 
        ($\leftarrow$) 
        If a view tuple $v$ is deleted, then all witnesses contributing to it are deleted. 
        ($\rightarrow$) 
        If a witness $w$ is not deleted, then any view tuple $v \subseteq w$ is not deleted.
        \begin{equation*}
            X[w] \geq X[v], v \subseteq w
        \end{equation*}
    \end{itemize}

\subsubsection{Naive ILP}
We define $\gdpilpbasic$ as the program resulting from our definitions of 
ILP variables, objective function, and constraints, 
and will sometimes refer to it as the ``naive ILP''.

\begin{restatable}{theorem}{thmbasicilp}[Naive ILP]
\label{th:ILP:basic}
    The interventions given by an optimum solution of $\gdpilpbasic$ for any 
    $\D$, $\querydel$, $\querypres$, $\querymin$, $\querymax$, $\kdel$, $\kpres$
    are an optimum solution to $\gdp$ over the same input. 

\end{restatable}

The direct mapping from the variables, objective and constraints of $\gdp$ into our ILP formulation from this section forms the proof.

\subsection{Wildcard Semantics for $X[w]$ and $X[v]$}
\label{sec:wildcard}
The binary variables for each input tuple $X[t]$ are always faithful to deletions in the database $\D$ (a tuple is either deleted or present). 
However, for witness variables $X[w]$ and output tuple variables $X[v]$ 
we use a semantics that we call ``\emph{wildcard semantics}.''
The intuition is that user constraints on deletion views provide hard lower bounds on deletions in the database (we need to provide at least that many deletions), 
while minimization views provide upper bounds (more deletions than necessary get automatically penalized by the optimization objective).
This results in a one-sided guarantee.
For example, setting $X[v_1]=1$ for $v_1 \in \querydelparam{i}(\D)$ means it is necessarily deleted,
and setting $X[v_2]=0$ for $v_2 \in \querypresparam{i}(\D)$ means it is necessarily preserved.
However, in this semantics we cannot infer the actual status from $X[v_1]=0$ and $X[v_2]=1$.
This semantics allows us to simplify the ILP by having fewer constraints; 
and, it turns out to be crucial for the tractability proofs in \cref{sec:tractibility-results}.

\begin{figure}
    \!\!\!\!\!
    \includegraphics[scale=0.37]{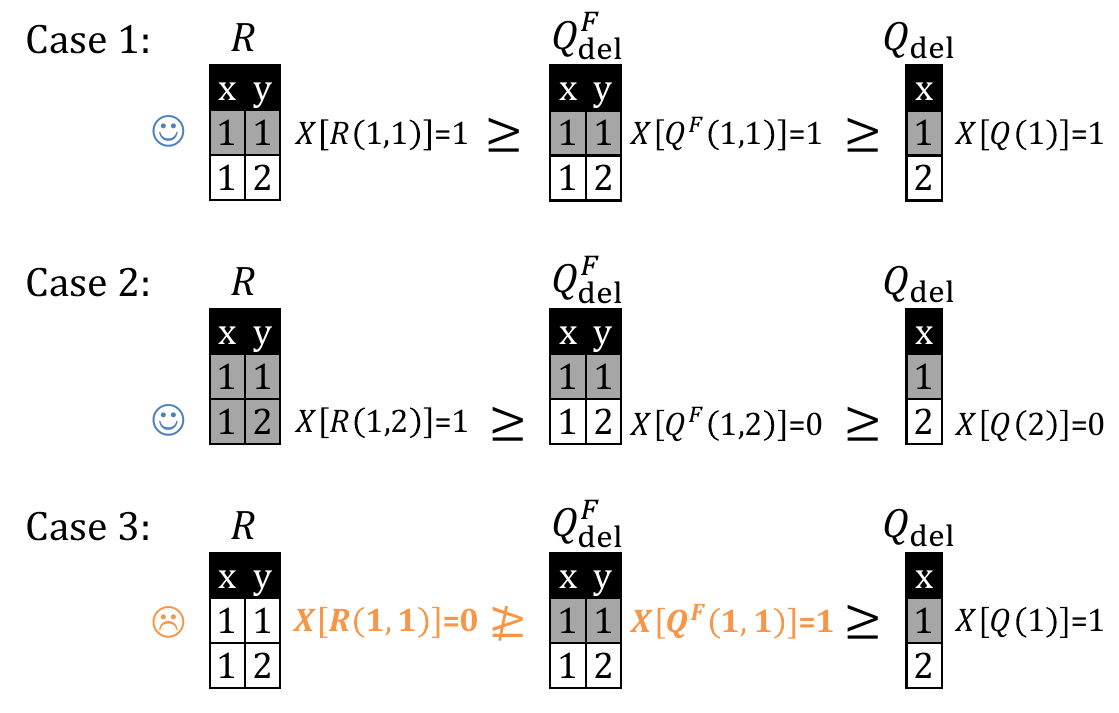}
    \caption{\Cref{ex:dontcaresemantics}: 
    The true state of deletions in the database $\D$ is always faithfully represented by database variables (e.g., $R(1,2)$ is deleted and thus $X[R(1,2)]=1$ and is grayed out).
    However, deletions in the views defined by a query in $\querydel$ need to provide only lower bounds for modeling $\dpss$ 
    (e.g., setting $X[Q^F(1,2)]=0$ in case 2 is ok even though the view tuple would be deleted).    }
    \label{Fig_Wildcard_semantics_one_direction}
\end{figure}

\begin{example}[wildcard semantics]
    \label{ex:dontcaresemantics}
    Consider a database $\D$ with facts $\{R(1,1), R(2,2), S(1)\}$, 
    and query 
    $Q(x) \datarule \allowbreak R(x,y), \allowbreak S(y)$. 
    Consider a $\dpss$ problem where tuple $Q{}(1)$ 
    should be deleted from the output.
    We introduce the tuple variables 
    $X[R(1,1)]$, 
    $X[R(1,2)]$, 
    $X[S(1)]$,
    witness variables 
    $X[Q^F(1,1)]$, 
    $X[Q^F(1,2)]$, 
    and view variables 
    $X[Q(1)]$, 
    $X[Q(2)]$.
    We show in \cref{Fig_Wildcard_semantics_one_direction} some possible variable assignments and discuss if they satisfy the wildcard semantics.

    Case 1: A feasible solution is setting
    $X[R(1,1)]$, $X[Q^F(1,1)]$, $X[Q(1)]$ to 1,
    and all other variables to 0, 
    i.e.\ tuple $R(1,1)$ is deleted from the database, 
    witness $Q^F(1,1)$ is deleted from the full query, and tuple $Q(1)$ is deleted from the view.
    In this case, all variables are faithful to a set of actual deletions in the database and views.

    Case 2: Another solution modifies $X[R(1,2)]$ to 1, 
    while the other variables remain the same (including $X[Q^F(1,2)]=0$).
    Since the witness $Q^F(1,2)$ would be deleted once $R(1,2)$ is deleted,
    this solution is not faithful to any set of interventions 
    (if 0 assignments are interpreted as required preservations).
    Notice, however, that this variable assignment causes no harm in the correct fulfillment of the user constraints.
    Marking a witness as not deleted when it is, is not a problem, since this can never mark the user constraint as satisfied if it isn't in reality.

    Case 3: In contrast, a solution with $X[R(1,1)]=0$, $X[Q^F(1,1)]=1$, $X[Q(1)]=1$ is incorrect. 
    It falsely claims to satisfy the user constraint by deleting the tuple $Q(1)$ from the view, 
    but it does not actually delete any input tuples that would lead to this deletion. 
\end{example}

\Cref{ex:dontcaresemantics} showed that for a variable $X[w]$ for a witness $w$ in $\querydel^F(\D)$, 
it is important that we do not claim it is deleted if it is not (as this would not truly satisfy the user requirement).
Thus, $X[w] = 1$ must imply that $w \not\in \querydel^F(\D)$.
However, deleting $w$ while having $X[w] = 0$  is not a problem, because this can never represent an unsatisfactory interventions (as is the case when a user required deletions that are not truly carried out).
Thus, we use a semantics for witnesses in $\querydel^F$ 
where $X[w]=1$ implies witness $w$ is deleted, 
while $X[w]=0$ acts as a ``\emph{wildcard}'', allowing the witness to be deleted or not. 
In other words, truth assignments to tuples in $\querydel(\D)$ provide a lower bound on the deletions of tuples in the database (\cref{Fig_Wildcard_semantics_one_direction}).
The exact same reasoning applies to the $X[v]$ variables for view tuples in $\querydel(\D)$ as well.

Similarly, witness and view variables for $\querymax^F$ and $\querymax$ provide upper bounds on tuple deletions in the database.
For these views, a solution 
stating that a witness / view tuples
is not deleted when it is, is not a problem since the user constraints specify a lower bound.
Thus, here too we allow the same semantics that $X[w]=1$ and $X[v]=1$ if $w$ / $v$ is deleted, and $X[w]=0$ and $X[v]=0$ are wildcard values where the witness/view tuple may or may not be deleted.

Symmetrically, for $\querypres(\D)$ and $\querymin(\D)$, 
we need to ensure that if a tuples and witnesses is said to be \emph{preserved}
(i.e.\ their variables are set to 0), then it is actually preserved. 
Thus, $X[w]=0$ and $X[v]=0$ for $w$, $v$ in $\querypres$ and $\querymin$ imply that the corresponding witness or view tuple is not deleted, while $X[w]=1$ and $X[v]=1$ represent a wildcard value, allowing the corresponding witness or view tuple to be deleted or not.
\Cref{tab:dontcaresemantics} captures the semantics of the $X[w]$ and $X[v]$ variables for each type of query.

\begin{figure}
    \setlength{\tabcolsep}{0.8pt}
    \centering
    \begin{tabular}{|c|c|c|c|c|}
        \hline
        & X[w] = 0 & X[w] = 1 & X[v] = 0 & X[v] = 1 \\
        \hline
        $\querymax(\D)$ & $*$ & $w \!\not\in\! \querymax^F(\D)$ & $*$ & $v \!\not\in\! \querymax(\D)$  \\
        $\querydel(\D)$ & $*$ & $w \! \not\in \! \querydel^F(\D)$ & $*$ & $v \!\not\in\! \querydel(\D)$  \\
        $\querypres(\D)$& $w \!\in\! \querypres^F(\D)$  & $*$ & $v\! \in\! \querypres(\D)$ & $*$ \\
        $\querymin(\D)$ & $w \!\in\! \querymin^F(\D)$ & $*$ & $v \!\in\! \querymin(\D)$ & $*$  \\
        \hline
    \end{tabular}
    \caption{Table showing the one-sided guarantees that any variable assignment has on solution to a $\gdp$ problem. For cases with wildcards (``$*$''), the true value of the variable can be either 0 or 1.}
    \label{tab:dontcaresemantics}
\end{figure}

\introparagraph{Selective application of PCs}
We use the wildcard semantics for witnesses and view variables described in \cref{sec:ilpvariables} to obtain a more efficient ILP.
Concretely, we don't apply the PCs in directions that are not required to enforce the wildcard semantics.

PC1 and PC3 encode lower bounds on the witness and view variables, respectively. 
They ensure that $X[w]=0$ and $X[v]=0$ only when $w$ and $v$ are not deleted. 
Thus, they need to be applied to $\querypres$ and $\querymin$,
but do not to $\querydel$ and $\querymax$.
Similarly, PC2 and PC4 are upper bounds on the witness and view variables, respectively.
They ensure that $X[w]=1$ and $X[v]=1$ only when $w$ and $v$ are deleted. 
Thus, they need to be applied to $\querydel$ and $\querymax$,
but not to $\querypres$ and $\querymin$.
\Cref{Fig_Wildcard_semantics_crossroad} summarizes the selective application of PCs to the different views.

\introparagraph{Wildcard ILP}
We refer to the ``wildcard ILP'' or $\gdpilpwildcard$ solution to $\gdp$ as the basic ILP that applies the PCs only selectively, 
namely PC1 and PC3 to 
$\querypres$ and $\querymin$
(but not PC2 nor PC4),
and PC2 and PC4 to 
$\querydel$ and $\querymax$
(but not PC1 nor PC3).

\begin{restatable}{theorem}{thmwildcardilp}[Wildcard ILP]
\label{th:ILP:wildcard}
    The interventions suggested by an optimum solution of $\gdpilpwildcard$ are an optimum solution to $\gdp$ over the same input. 
\end{restatable}

\begin{proofintuition*}
    The proof is based on the fact that any optimal solution under traditional semantics is an optimal solution in the wildcard semantics, and to enforce the wildcard semantics it suffices to apply PCs selectively (which is possible as argued before).    
\end{proofintuition*}

\begin{figure}
    \centering
    \includegraphics[scale=0.35]{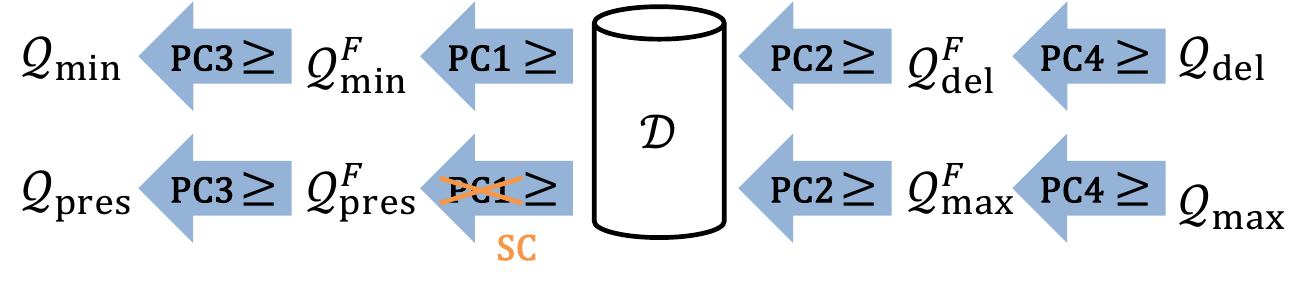}
    \caption{
        Arrows in this figure illustrate the constraints \emph{in the direction of lower bounds} 
        (but recall that constraints are bidirectional).
        Notice that our wildcard semantics applies constraints only selectively to different views 
        (\cref{sec:wildcard}).
        Also shown is how our Smoothing Constraints (SC) replace PC1 for 
        $\querypres^F(\D)$ (\cref{sec:smoothing}).
        It is that replacement that gives us a powerful PTIME guarantee for PTIME problems
        (see later \cref{fig:expt-3} from the experiments).
    }
    \label{Fig_Wildcard_semantics_crossroad}
\end{figure}

\subsection{ILP with Smoothing Constraints} 
\label{sec:smoothing}

The wildcard semantics alone does not give noticeable performance improvements or PTIME guarantees for our ILP.
However, it allows us to enable a surprising optimization: we will tighten one type of constraint in a way that the resulting solution space (a polyhedron) preserves an optimal solution, yet also affords desirable properties on the performance of the resulting ILP and also the optimal solution for its LP relaxation.
It is those seemingly superfluous constraints that play a key ingredient in the results of \cref{sec:tractibility-results} where we show that an ILP with smoothing constraints $\gdpilpsmooth$ can be solved in PTIME for all known tractable cases. 
Hence, we also refer to $\gdpilpsmooth$ as simply $\gdpilp$.

The user constraints and propagation constraints suffice to correctly model GDP as an ILP.
The purpose of the \emph{Smoothing Constraints} (SC) is to make the objective of the LP relaxation closer to the objective of the ILP (in certain cases we see that the smoothing constraint makes the optimal objective value of LP relaxation equal to that of the ILP).
In the language of linear optimization, adding these extra bounds is equivalent to adding cutting planes~\cite{Kelley1960CuttingPlane} to the polytope defined by the LP relaxation. 

We identify a smoothing constraint that can be added to describe the relation between tuple variables, and the witness variables of $\querypres$. 
Recall that for $\querypres$, we would like to preserve a certain number of view variables. 
A view variable $v$ is preserved if at least one of its witnesses $w$ is preserved. 
Recall that due to our wildcard semantics of $X[w]$ in $\querypres$, 
setting $X[w]=1$ means that we ``do not care'' whether the witness is deleted or not. 
In other words, we can say that for any view variable $v$, there is only one $w$ with $X[w]=0$ and the other witnesses can be set to 1.
Now assume that a tuple $t$ participates in multiple witnesses $w_1, w_2, \ldots, w_k$ corresponding to the same view tuple $v$
(It may also participate in more witnesses, but we do not care about those here).
We know through PC1 that $X[t] \leq X[w]$ for a given $v$ and $t \in w, w \supseteq v$. 

It is correct to now also enforce that $\sum_{i \in [1,k]} X[w_i] \geq k-1$
i.e., only one $X[w]$ in this set is preserved (the rest may also be preserved, but due to the wildcard semantics they will still have $X[w]=1$).
Now we can also enforce that $X[t] \leq (\sum_{i \in [1,k]} X[w_i]) - (k-1)$, since  
all but $1$ values of $X[w]$ are set to $1$, and only the final value decides the upper bound on $X[t]$. 
Thus, we get a smoothing constraint, applied to every $v \in \querypres(\D)$:
\begin{align*}
    X[t] \leq 1 +
        \sum_{\substack{w: t\in w \\ w \supseteq v}} 
        ( X[w] -  1) 
\end{align*}

\introparagraph{Correctness of $\gdpilpsmooth$ with wildcard semantics and smoothing constraints}
We refer to $\gdpilpsmooth$ as the ILP that has only the PCs that are required for wildcard semantics and has replaced the PC1 constraints
on $\querypres$ in the basic ILP with SC instead.

\begin{restatable}{theorem}{thmsmoothenedilp}[Smoothened ILP]
\label{th:ILP:smoothened}
    The interventions suggested by an optimum solution of $\gdpilpsmooth$ 
    with wildcard semantics and smoothened constraints
    form an optimum solution to $\gdp$.
\end{restatable}

\begin{proofintuition*}
    Adding the smoothing constraint to the wildcard ILP always preserves at least one optimal solution - this follows also from the argument above the the smoothing constraint can be derived by logically following the wildcard semantics.
\end{proofintuition*}

\introparagraph{An interesting asymmetry}
We notice an interesting asymmetry at play.
We could apply a symmetric smoothing constraint in the case for PC4 on $\querydel$.
Interestingly, such an additional smoothing constraint would \emph{identical to our original PC1}, 
as every view tuple corresponds to exactly one witness.
Thus, we do not need to add any additional smoothing constraints for $\querydel$.

\introparagraph{Reducing the size of the ILP}
The smoothing constraints may subsume some propagation constraints. 
These subsumed propagation constraints can be removed from the ILP without affecting any solution of the ILP or LP relaxation.

\introparagraph{The Power of Smoothing Constraints}
\cref{ex:smoothing} is an intuitive example of a $\swp$ problem instance modelled as a $\gdp$ problem,
where the smoothing constraints ensure that the optimal value of the ILP is equivalent to the optimal value of its LP relaxation in the $\gdp$ framework.
Later in \cref{prop:swpptime}, we show that this is the case for all prior known PTIME cases of $\swp$.

\begin{example}[Power of smoothing]
\label{ex:smoothing}    

Consider again the $\D$ from \cref{ex:dontcaresemantics}: with $R(1,1)$, $R(1,2)$, and $S(1)$.
We want to solve the smallest witness problem 
$\swp(\querypresparam{}, \D)$ for 
$\querypresparam{}(x) \datarule R(x,y), S(x)$. 
To model it as $\gdp$, we set
$\querypres$ to be $\langle \querypresparam{} \rangle$ and 
$\kpres = \langle \kpresparam{} \rangle$
with
$\kpresparam{}= 1$, which is the number of output tuples in 
$\querypresparam{}(\D)$.
We also set 
$\querymax = \langle 
    \querymaxparam{1}(x,y)\datarule$ $R(x,y), 
    \querymaxparam{2}(x)\datarule S(x)
\rangle$
and $\querydel = \querymin = \emptyset$.
Our $\gdp$ formulation is as follows:
\begin{align*}
    &f(\vec X) = - \big( X[\querymaxparam{1}(1,1)] + X[\querymaxparam{1}(1,2)] + X[\querymaxparam{2}(1)] \big)
\end{align*} 
s.t.\ following constraints (and integrality constraints):
\allowdisplaybreaks
\begin{align*}        
    & X[\querypresparam{}(1)] \leq 0  
        &\textrm{(UC)} \\
    & X[\querypresparam{F}(1,1)] + X[\querypresparam{F}(1,2)] - 1 \leq X[\querypresparam{}(1)] 
    &\textrm{(PC3)}\\ 
    & X[R(1,1)] \leq X[\querypresparam{F}(1,1)]     &\textrm{(PC1)}\\ 
    & X[S(1)] \leq X[\querypresparam{F}(1,1)]     &\textrm{(PC1)}\\ 
    & X[R(1,2)] \leq X[\querypresparam{F}(1,2)]     &\textrm{(PC1)}\\ 
    & X[S(1)] \leq X[\querypresparam{F}(1,2)]     &\textrm{(PC1)}\\ 
    & X[R(1,1)] \leq X[\querymaxparam{1F}(1,1)]     &\textrm{(PC2)}\\ 
    & X[R(1,2)] \leq X[\querymaxparam{1F}(1,2)]     &\textrm{(PC2)}\\ 
    & X[S(1)] \leq X[\querymaxparam{2F}(1)]     &\textrm{(PC2)}\\ 
    & X[\querymaxparam{1F}(1,1)] \leq X[\querymaxparam{1}(1,1)]    &\textrm{(PC4)}\\ 
    & X[\querymaxparam{1F}(1,2)] \leq X[\querymaxparam{1}(1,2)]   &\textrm{(PC4)}\\ 
    & X[\querymaxparam{2F}(1)] \leq X[\querymaxparam{2}(1)]    &\textrm{(PC4)}
\end{align*}

Observe that the optimal solution for the ILP is $-1$ which occurs when either one of the tuples in $R$ is deleted, i.e.\ either of $X[R(1,1)]$ or $X[R(1,2)]$ is set to $1$.
However, the LP relaxation has a smaller non-integral optimum of $-1.5$ for $X[R(1,1)] = X[R(1,2)] = X[S(1)] = 0.5$.
This is due to the fact that both $X[\querypresparam{F}(1,1)]$
and $X[\querypresparam{F}(1,2)]$ can take values $0.5$, which is why 
$X[\querypresparam{1}(1)]$ can be set to $0$ while fulfilling all constraints.

Our smoothing constraint for this example is the following 
\begin{align*}
    &X[S(1)] \leq X[\querypresparam{F}(1,1)] + X[\querypresparam{F}(1,2)] - 1 \qquad
        &\textrm{(SC)}
\end{align*}
\noindent
Notice that it can replace the $PC1$ constraints $X[S(1)] $$\leq$$ X[\querypresparam{F}(1,1)]$ and $X[S(1)] $$\leq$$ X[\querypresparam{F}(1,2)] $, since it is a strictly tighter constraint.

The SC ensures that if $ X[S(1)]$ is set to $0.5$, 
then $X[\querypresparam{F}(1,1)]+ X[\querypresparam{F}(1,2)] \geq 1.5$, thus violating PC4, and thereby effectively removing the non-integer solution.
To gain more intuition, \cref{Fig:smoothing:illustration} shows the polytope of the LP relaxation of our wildcard formulation projected on
either the variables involved in SC (\cref{Fig_cuttingplane_2}),
or the three input tuples (\cref{Fig_cuttingplane_1}).
The optimal LP solution corresponds to the orange point 
(point 3 in \cref{Fig_cuttingplane_2}, point 6 in \cref{Fig_cuttingplane_1}),
and the two optimal ILP solutions correspond to the two blue points 
(points 1 and 2 in \cref{Fig_cuttingplane_2}, points 0 and 2 in \cref{Fig_cuttingplane_1}).
Notice how our SC (shown as yellow cutting plane in \cref{Fig_cuttingplane_2}) cuts away the non-integer solution,
leaving only points 1 and 2, and their convex extension.
Similarly, this constraint cuts away all points with $X[S(1)]>1$ (not shown in \cref{Fig_cuttingplane_1}),
leaving points 0 and 2 and their convex combination as the optimum solutions to the new LP.

\end{example}

\begin{figure}
    \centering
    \begin{subfigure}{0.5\columnwidth}
        \centering
        \includegraphics[scale=0.42]{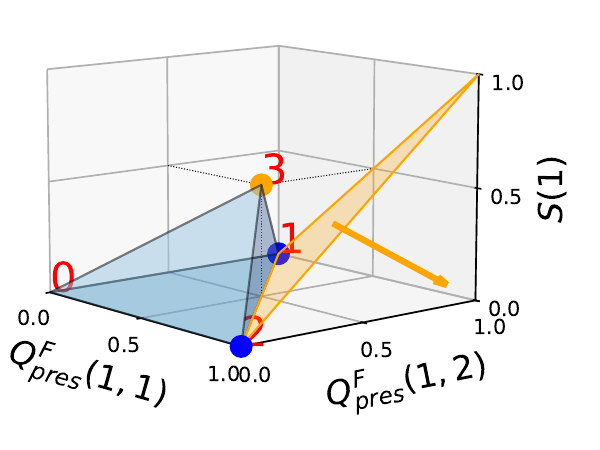}
        \caption{Projection on variables from SC.
        }
        \label{Fig_cuttingplane_2}
    \end{subfigure}
    \begin{subfigure}{0.48\columnwidth}
        \centering
        \includegraphics[scale=0.42]{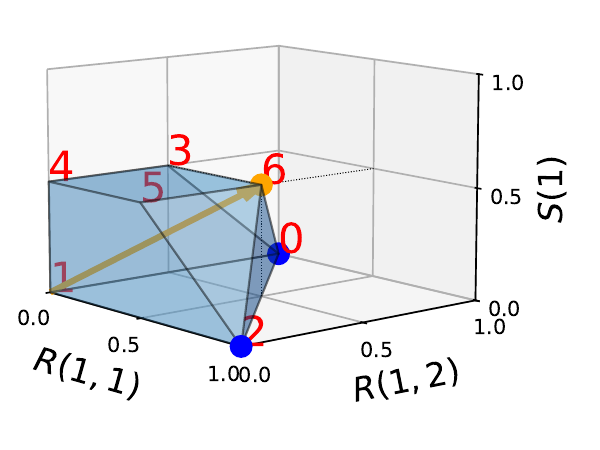}
        \caption{Projection on database tuples.
        }
        \label{Fig_cuttingplane_1}
    \end{subfigure}%

    \caption{\cref{ex:smoothing}: Our Smoothing Constraint (SC) acts as a cutting plane, removing a non-integral optimal point from the LP relaxation of our ILP formulation (see details in text).
    }
    \label{Fig:smoothing:illustration}
\end{figure}

From the workings of modern solvers we know that any ILP problem can be solved efficiently if 
its natural LP relaxation is tight with the ILP polytope in the direction of the objective (i.e.\ the ILP and its LP relaxation have the same optimal $f^*$ and share the same ``face'' perpendicular to the objective vector).
Now, it suffices to show that the LP relaxation of the smoothened ILP has the same optimum objective values $f^*$ and preserves at least one optimal integral solution.
We see experimentally in \cref{fig:expt-3} in \cref{sec:expts} a speedup of $2$ orders of magnitude in the ILP solving time simply by adding the smoothing constraints.
This is completely justified by our claim that ILP solvers are able to solve ILPs efficiently when the LP relaxation is tight.

\section{Recovering Existing Tractability Results}
\label{SEC:TRACTABILITY}
\label{sec:tractibility-results}

In this section we focus on self-join free queries under set semantics, which is the only case in which complexity dichotomies are known for the DP variants of $\dpss, \dpvs, \adp, \swp$.
We have shown previously that the $\gdp$ framework naturally captures all these problems as special cases. 
In this section, we show that the LP relaxation of the $\gdp$ problem also naturally recovers the optimal, integral solutions for these problems for self-join queries that are known to be tractable under set semantics.

\introparagraph{DP-SS} A $\dpss$ problem on a query $Q$ can be converted to a resilience problem on the existential version of the query $Q^E$ which 
is obtained by removing all head variables from $Q$ (both in the head and the body).
Since a dichotomy result for self-join free conjunctive queries is known for resilience both under set and bag semantics, it follows that a complexity dichotomy is also known for $\dpss$.

Let $\gdpdpssilpparam{Q, \D, t}$ be the ILP obtained when we pose the $\dpss$ problem over a query $Q$, database $\D$ and target tuple in view $t$, as a $\gdp$ problem via the method described in \cref{sec:gdp-dp-ss}. 
We now claim that the LP relaxation $\gdpdpsslp$ of such an ILP $\gdpdpssilp$, is always equivalent to the solution of the optimization problem $\dpss$ for all known queries $Q$ for which $\dpss$ can be solved in PTIME, and thus $\gdpilp$ can be used to solve $\dpss$ in PTIME for such queries.

\begin{restatable}{proposition}{dpssptime}
    $\gdpdpsslpparam{Q, \D} = \dpss(Q, \D)$ for all database instances $\D$ under set semantics if the existential query $Q^E$ does not contain a triad. 
     $\gdpdpsslpparam{Q, \D} = \dpss(Q, \D)$ for all database instances $\D$ under bag semantics if $Q^E$ is a linear query.
     \label{prop:dpssptime}
\end{restatable}

\begin{proofintuition*}
We show that $\gdpdpssilp$ is identical to a specialized ILP that has been proposed \cite{makhija2024unified} for resilience. 
Since that paper also shows that for all tractable queries, the LP relaxation of the ILP is integral, the results naturally carry over.
\end{proofintuition*}

\introparagraph{DP-VS}
It is known that DP-VS is PTIME for self-join free conjunctive queries if and only if they have the head domination property \cite{KimelfeldVW12}. 
We prove that for such queries that have the head domination property, if we pose $\dpvs(Q, \D, t)$ in the $\gdp$ framework, then the LP relaxation $\gdpdpvslpparam{Q, \D, t}$ is equivalent to the solution of $\dpvs(Q, \D, t)$.

\begin{restatable}{proposition}{dpvsptime}
    $\gdpdpvslpparam{Q, \D, t} = \dpvs(Q, \D, t)$ for all database instances $\D$ under set semantics and any tuple $v$ in $Q(D)$ if $Q$ has the head domination property.
    \label{prop:dpvsptime}
\end{restatable}
\begin{proofintuition*}
    If a query has the head domination property, it is known that optimal solution for $\dpvs$ is side effect free~\cite{KimelfeldVW12}. 
    Thus, the optimal value of the ILP objective is $0$, and LP relaxation cannot take on a negative value and hence must be equal and integral.
\end{proofintuition*}

\introparagraph{SWP} 
It is known that SWP is PTIME for self-join free queries if and only if they have the \emph{head clustering property} \cite{hu_et_al:LIPIcs.ICDT.2024.24}, which is a restriction of the head domination property.
We are again able to show that for such queries that have the head clustering property, if we pose $\swp(Q, \D)$ in the $\gdp$ framework, then the LP relaxation $\gdpswplpparam{Q, \D}$ is equivalent to the solution of $\swp(Q, \D)$.

\begin{restatable}{proposition}{swpptime}
    $\gdpswplpparam{Q, \D} = \swp(Q, \D)$ for all database instances $\D$ under set and bag semantics if $Q$ is a self-join free conjunctive query with the head clustering property.
    \label{prop:swpptime}
\end{restatable}

\begin{proofintuition*}
    We first simplify the ILP and phrase it in terms of variables $Y[t] = 1 - X[t]$. 
    We next show that due to the head clustering property, the ILP can be decomposed into multiple independent ILPs, corresponding to different existentially connected components of the query.
    For each such component, the correct solution can be obtained by preserving an arbitrary witness for each projection, and hence the LP relaxation must be tight.
\end{proofintuition*}

\introparagraph{ADP-SS} 
A complexity dichotomy for the $\adp$ problem for self-join free queries under set semantics is known. 
However, the complexity criterion \cite{ADP} is much more involved.
In particular, $\adp$ for a self-join free query is PTIME if and only if 
(1) The query is Boolean and does not have a triad, 
(2) The query has a \emph{singleton} relation, 
(3) Repeated application of decomposition by removing head variables that are present in all atoms, and treating disconnected components of a query independently, results in queries that are tractable.
We show that no matter the reason for tractability, if we pose $\adp(Q, \D, k)$ in the $\gdp$ framework, then the LP relaxation $\gdpadplp$ is equivalent to the solution of $\adp(Q, \D, k)$.

\begin{restatable}{proposition}{adpptime}
    $\gdpadplpparam{Q, \D} = \swp(Q, \D)$ for all database instances $\D$ if $Q$ is a self-join free for which $\adp(Q)$ is known to be tractable under set semantics.
    \label{prop:adpptime}
\end{restatable}

\begin{proofintuition*}
    The proof of optimality of the LP relaxation for $\adp$ is similar to the proof for tractability in $\adp$\cite{ADP} in terms of the base cases and how queries are decomposed. 
    For the base case of boolean queries without a triad, we use the proof of \cref{prop:dpssptime} as an argument, while for the base case of singleton relations, we use the proof similar to that of \cref{prop:dpvsptime}.
    We also show that the value of the LP relaxation is preserved even when the query is decomposed into multiple parts.
\end{proofintuition*}

\section{Experiments}
\label{sec:expts}

The goal of our 
experiments is to evaluate the performance of our unified $\gdpilp$ 
(which is our short form for $\gdpilpsmooth$) by answering the following 4 questions: 
(Q1)~Is the performance of $\gdpilp$ comparable to previously proposed specialized algorithms tailored to PTIME cases of particular DP problems? 
(Q2) Can our unified $\gdpilp$ indeed efficiently solve new tractable cases  with self-joins, unions, and bag semantics that 
we prove to be in PTIME in \cref{SEC:TRACTABILITYNEW}?
(Q3) What, if any, is the performance benefit we obtain via smoothing constraints as discussed in \cref{sec:ilp-framework}?
(Q4) What is the scalability of solving completely novel DP problems that fall into our unified $\gdp$ framework on real-world data?

\introparagraph{Algorithms}  
$\gdpilp$ denotes our ILP formulation for the generalized deletion propagation. 
$\dpvsspl$, 
$\swpspl$, 
$\adpspl$ 
denote prior specialized algorithms (``-S'') for the three problems 
\circled{\footnotesize V} $\dpvs$ \cite{KimelfeldVW12}, 
\circled{\footnotesize S} $\swp$ \cite{hu_et_al:LIPIcs.ICDT.2024.24}, and 
\circled{\footnotesize A} $\adp$ \cite{ADP}, respectively. 
Recall that these are dedicated algorithms proposed for a PTIME cases of particular problems.
We were not able to find open source code for any of these problems and implemented them based on the pseudocode provided in the respective papers that proposed them:
\circled{\footnotesize V} \cite{KimelfeldVW12}, 
\circled{\footnotesize S} \cite{hu_et_al:LIPIcs.ICDT.2024.24},
\circled{\footnotesize A} \cite{ADP}. 
To the best of our knowledge, no experimental evaluation has ever been undertaken for some of these algorithms~\cite{hu_et_al:LIPIcs.ICDT.2024.24,KimelfeldVW12}.
We do not include experimental comparison for $\dpss$, as for this problem the $\gdpilp$ produced is exactly the same as a prior specialized approach \cite{makhija2024unified}, and hence there is no difference in performance.

\introparagraph{Data} 
For most experiments we generate synthetic data by
fixing the max domain size to 1000, and sampling randomly from all possible tuples. 
For experiments under bag semantics, each tuple is duplicated by a random number that is smaller than a pre-specified max bag size of 10.
For answering (4) in \cref{fig:expt-4}, we use an existing flights' database~\cite{DVN/WTZS4K_2020} that shows flights operated by different airlines in Jan 2019 as real world data case study on a novel problem.

\introparagraph{Software and Hardware}
The algorithms are implemented in Python 3.8.8 and solve the optimization problems with Gurobi Optimizer 10.0.1. 
Experiments are run on an Intel Xeon E5-2680v4 @2.40GH machine available via the Northeastern Discovery Cluster.

\introparagraph{Experimental Protocol} 
For each plot we run 3 runs of logarithmically and monotonically increasing database instances. 
We plot all obtained data points with a low saturation, and draw a trend line between the median points from logarithmically increasing sized buckets. 
All plots are log-log, and we include a dashed line to show linear scalability as reference in the log-log plot.

\begin{figure}[t]
    \centering
    \begin{subfigure}[b]{0.04\textwidth}
    \end{subfigure}
    \hfill
    \begin{subfigure}[b]{0.45\columnwidth}
        \centering
        \vfill
        \textbf{$\qthreeray$} \\
        (Known tractable Case)
        \vfill
    \end{subfigure}
    \begin{subfigure}[b]{0.45\columnwidth}
        \centering
        \vfill
        \textbf{$\qsjtriangletwochain$} \\
        {\crefname{section}{Sec}{Sections}
        (Tractable for $\dpvs$ and $\swp$)}
        \vfill
    \end{subfigure}

    \begin{subfigure}[b]{0.02\textwidth}
        \vfill
        \begin{turn}{90}\textbf{\qquad\qquad\quad\circled{\footnotesize\normalfont V}$\dpvs$}\end{turn}
        \vfill
    \end{subfigure}
    \begin{subfigure}[b]{0.47\columnwidth}
        \includegraphics[scale=0.35]{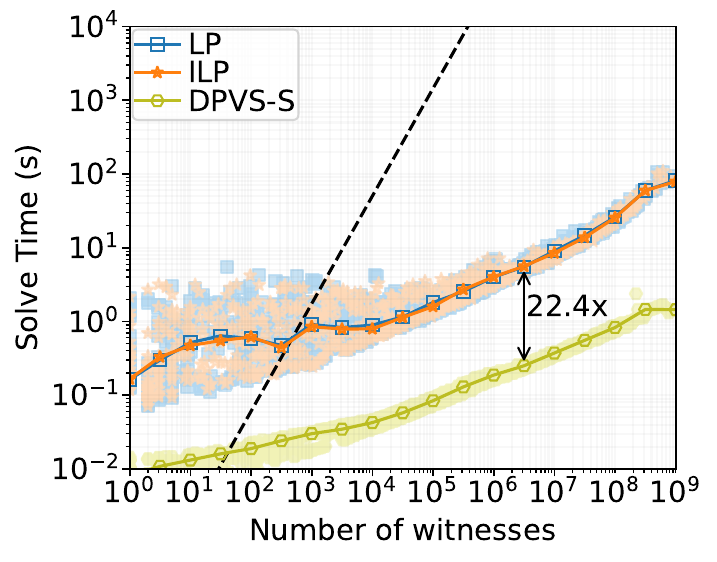}
        \caption{}
        \label{fig:dpvs-specialized}
    \end{subfigure}
    \begin{subfigure}[b]{0.47\columnwidth}
        \includegraphics[scale=0.35]{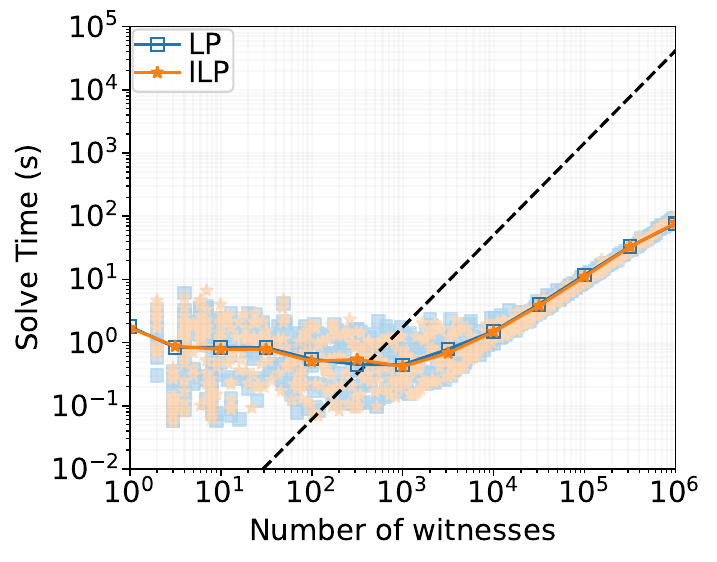}
        \caption{}
        \label{fig:sj-dpvs}
    \end{subfigure}
    \begin{subfigure}[b]{0.02\textwidth}
        \vfill
        \begin{turn}{90}\textbf{$\qquad\qquad\quad\circled{\footnotesize\normalfont S}\swp$}\end{turn}
        \vfill
    \end{subfigure}
    \begin{subfigure}[b]{0.47\columnwidth}
        \includegraphics[scale=0.35]{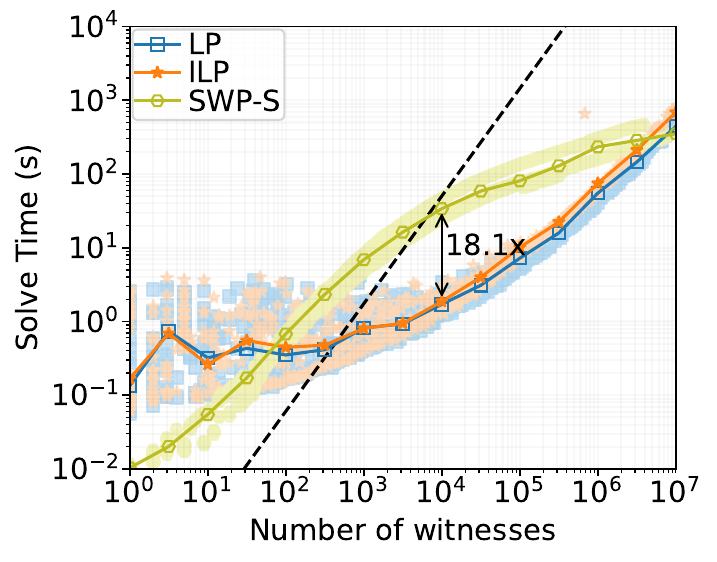}
        \caption{}
        \label{fig:swp-specialized}
    \end{subfigure}
    \begin{subfigure}[b]{0.47\columnwidth}
        \includegraphics[scale=0.35]{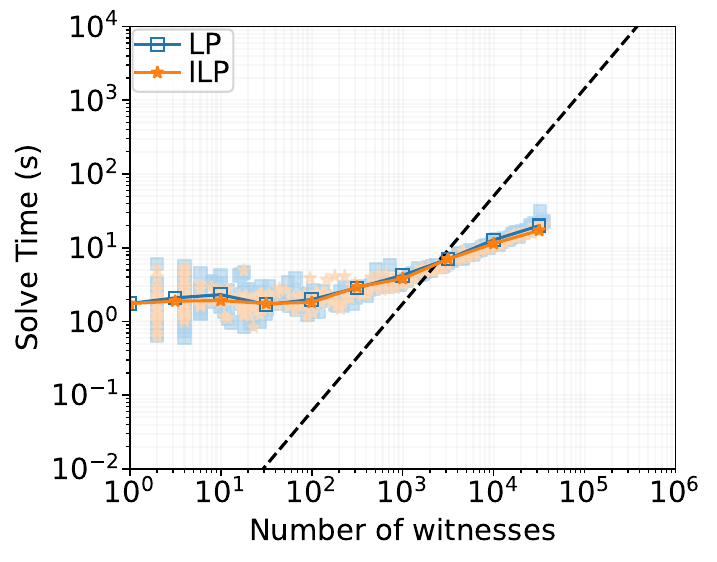}
        \caption{}
        \label{fig:sj-swp}
    \end{subfigure}
    \begin{subfigure}[b]{0.02\textwidth}
        \vfill
        \begin{turn}{90}\textbf{$\qquad\qquad\quad\circled{\footnotesize\normalfont A}\adp$}\end{turn}
        \vfill
    \end{subfigure}
    \begin{subfigure}[b]{0.47\columnwidth}
        \includegraphics[scale=0.35]{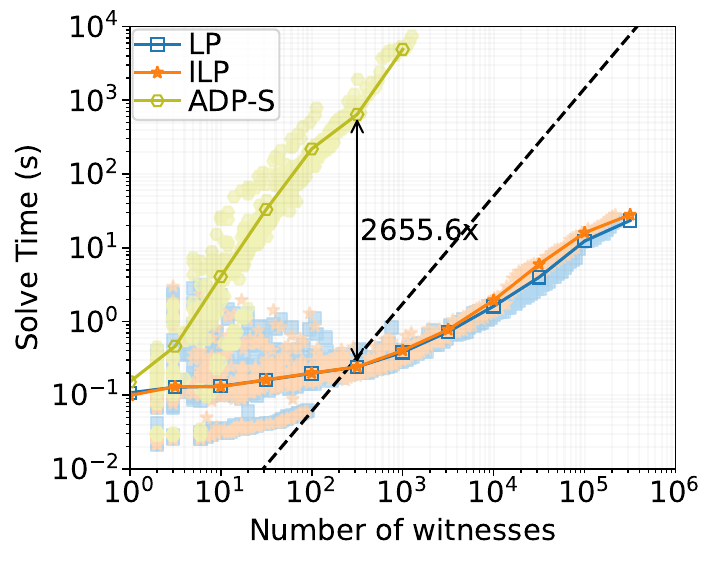}
        \caption{}
        \label{fig:adp-specialized}
    \end{subfigure}
    \begin{subfigure}[b]{0.47\columnwidth}
        \includegraphics[scale=0.35]{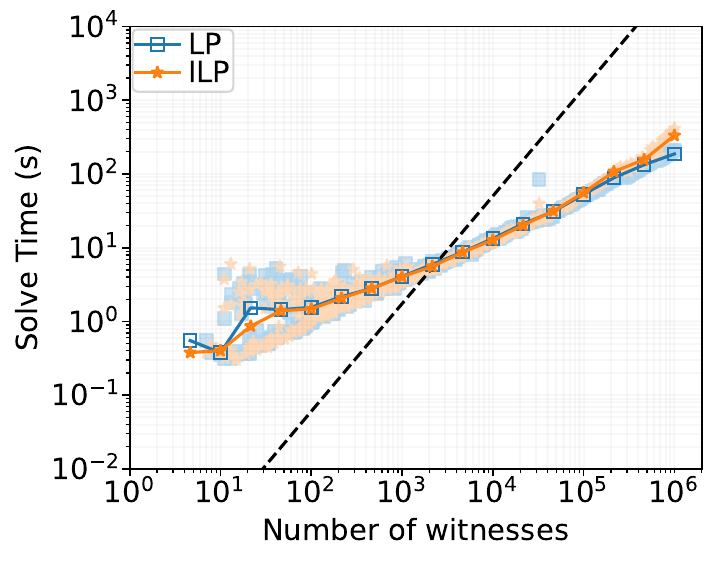}
        \caption{}
        \label{fig:sj-adp}
    \end{subfigure}
    \caption{
    (Q1)/(Q2): Performance of $\gdpilp$ on a previously known (left column compared against prior specialized algorithms) 
    and a newly discovered tractable query (right column with no prior known specialized algorithm) 
    for three prior studied problems 
    \protect\circled{\footnotesize\normalfont V}~$\dpvs$, 
    \protect\circled{\footnotesize\normalfont S}~$\swp$, and 
    \protect\circled{\footnotesize\normalfont A}~$\adp$. 
    In all cases, $\gdpilp$ scales well under the theoretical worst case complexity of ILPs and LPs. 
    }
    \label{fig:expt-12}
\end{figure}

\smallsection{(Q1) Known tractable cases}
Is the performance of $\gdpilp$ over PTIME instances comparable to specialized algorithms studied in prior work?
We pick the 3-star query $\qthreeray(a) \datarule R(a,b), S(a,c), R(a, d)$, for which all three problems can be solved in PTIME. 
We run all three problems on this query and compare the performance of $\gdpilp$ against specialized algorithms.
The ILP and the LP have worse worst-case complexity than the specialized algorithms, however we see that $\gdpilp$ is at times even faster than the specialized algorithms, due to the better heuristics used in the ILP solver. 
In \cref{fig:dpvs-specialized}, we see that $\gdpilp$ is about $20$ times worse than the specialized algorithms for \circled{\footnotesize D}~$\dpvs$. But this is expected for this case, since the PTIME cases of $\dpvs$ are only those where there are \emph{no side effects} and any tuple that contributes to the answer can be deleted, thus making the problem solvable via a trivial algorithm. 
However, in \cref{fig:swp-specialized,fig:adp-specialized}, we see that $\gdpilp$ is about $2$ and $3$ orders of magnitude \emph{faster than the specialized algorithms} for \circled{\footnotesize S} $\swp$ and \circled{\footnotesize A} $\adp$, both of which use decomposition based techniques, with the specialized algorithm for \circled{\footnotesize{A}} $\adp$ also requiring dynamic programming.
Thus, although the specialized algorithms have better asymptotic fine-grained complexity
and will perform better on adversarially chosen instances, over random instances the LP solver (using heuristics) is able to find a solution faster.

\smallsection{(Q2) Newly discovered Tractable cases}
We evaluate the performance of $\gdpdpvsilp$, $\gdpswpilp$ and $\gdpadpilp$ for $\qsjtriangletwochain$,
a query with self-joins and a union that we run under bag semantics.
We show \cref{SEC:TRACTABILITYNEW} that the \circled{\footnotesize D}~$\dpvs$, \circled{\footnotesize S}~$\swp$, and \circled{\footnotesize A}~$\adp$ problems are tractable for this query under bag semantics.
There are no known specialized algorithms for these problems, and thus we only compare the performance of $\gdpilp$ with the PTIME LP Relaxation.
We see in \cref{fig:sj-dpvs,fig:sj-swp,fig:sj-adp} that the ILP is as fast as its LP Relaxation and shows linear scalability even for this complicated setting.

\smallsection{(Q3) The Power of Smoothing Constraints}
\Cref{fig:expt-3} shows two orders of magnitude speedup in the ILP solving time after adding our Smoothing Constraint (SC).
We run $\gdpswpilp$ for the 3-star query $\qthreeray$,
contrasting $\gdpilp$ with and without smoothing constraints.
Since we show in \cref{sec:tractibility-results} that the LP relaxation of $\gdpilp$ shares the same optimum objective value,
this surprising speed-up completely justified by our the fact that ILP solvers can solve ILPs efficiently when the LP relaxation is tight.
Moreover, we see that the LP relaxation of $\gdpswpilp$ is always tight for the $\qthreeray$ query. 
This is notably not true for the naive ILP formulation, which can have an over $30\%$ higher optimal objective (GDP) value.\footnote{
    Due to $\swp$ being a maximization problem being posed as a minimization problem through $\gdpswp$, the optimal values of $\gdpswp$ are negative, and hence the magnitude of the LP relaxation is higher than the ILP (despite the LP being a lower bound).
}

\begin{figure}
    \centering
    \begin{subfigure}[b]{0.47\columnwidth}
        \includegraphics[scale=0.3]{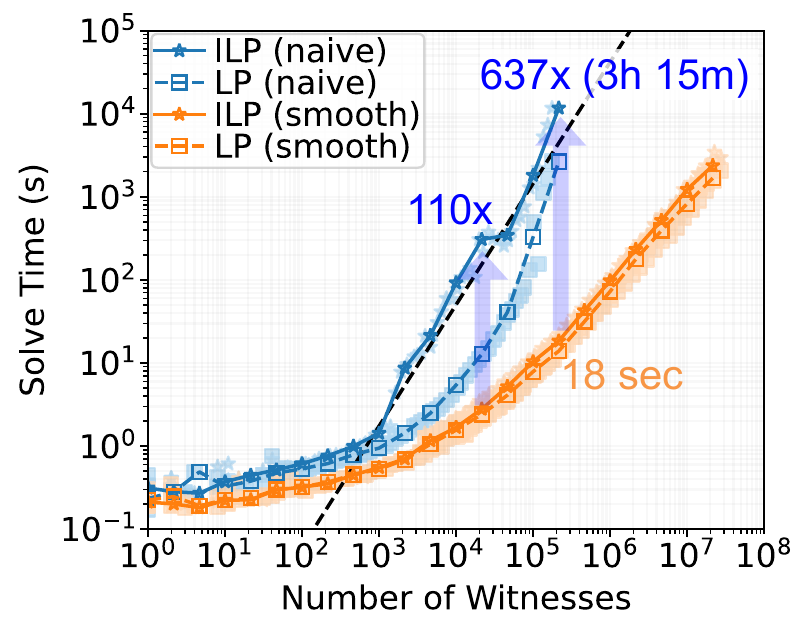}
    \end{subfigure}
    \begin{subfigure}[b]{0.47\columnwidth}
        \includegraphics[scale=0.33]{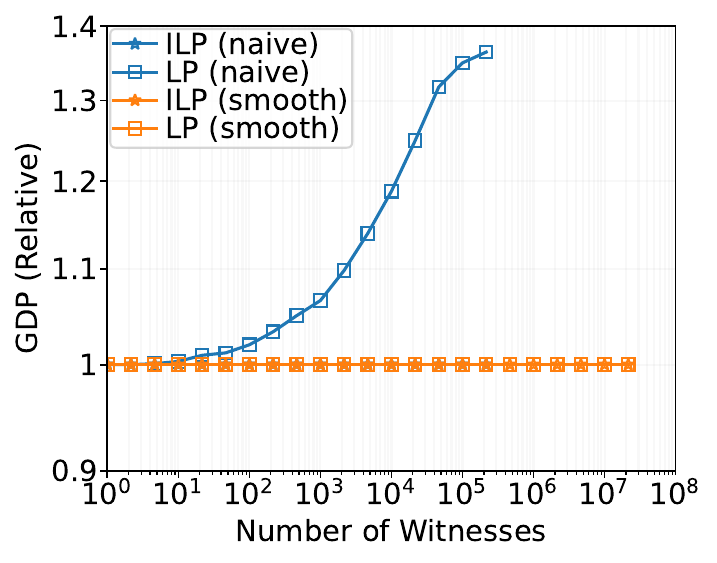}
    \end{subfigure}
    \caption{
    (Q3): Experiment showing the power of Smoothing Constraints in $\gdpilp$: we observe that
    $\gdpilp$ is orders-of-magnitude faster than the naive ILP formulation $\gdpilpbasic$, while also guaranteeing the optimality of its LP relaxation. Contrast with the LP relaxation of $\gdpilpbasic$ which can have an over $30\%$ higher optimal objective (GDP) value.}
    \label{fig:expt-3}
\end{figure}

\smallsection{(Q4) General performance}
We use the flights' database~\cite{DVN/WTZS4K_2020} that shows $500$K+ flights operated by $17$ different airlines in Jan 2019. 
We take the use case of \cref{ex:introexample} and solve the $\gdp$ with the following requirements: 
(1) Cut 2\% of total costs of the airline, 
(2) Let the connection network of an airline contain all pairs of places that have a direct (0-hop) or 1-hop flight between them. Minimize the effect of deletions on the connection network i.e., minimize the number of location pairs that are removed from the connection network.
(3)~A subset of pairs of places are ``popular connections''. Ensure that such connections are preserved in the connection network.
For our case study, we assign a random expense value to each airport as an airport fee and each flight as a fuel cost, since the dataset does not contain any cost information. 
However, such information or cost factors if known can be easily incorporated into the ILP formulation by simply changing the randomly assigned costs to the actual costs.
We assume if an airline uses an airport, it must pay its fee and assume that there are no additional costs.
We show in \cref{fig:expt-4} the time to solve the $\gdp$ for this use case via $\gdpilp$, and compare it to the LP relaxation.
We see that even for this real-world dataset, where there are no guarantees on the properties of $\gdpilp$, most instances are solved in well under a minute, and the optimal solutions of the ILP and LP relaxation coincide in many cases.\footnote{We observed that in some cases, ILP is faster than LP. This is a known observation and may be due to numerical and floating-point issues~\cite{gurobinum}.} 

\begin{figure}
    \centering
    \includegraphics[scale=0.35]{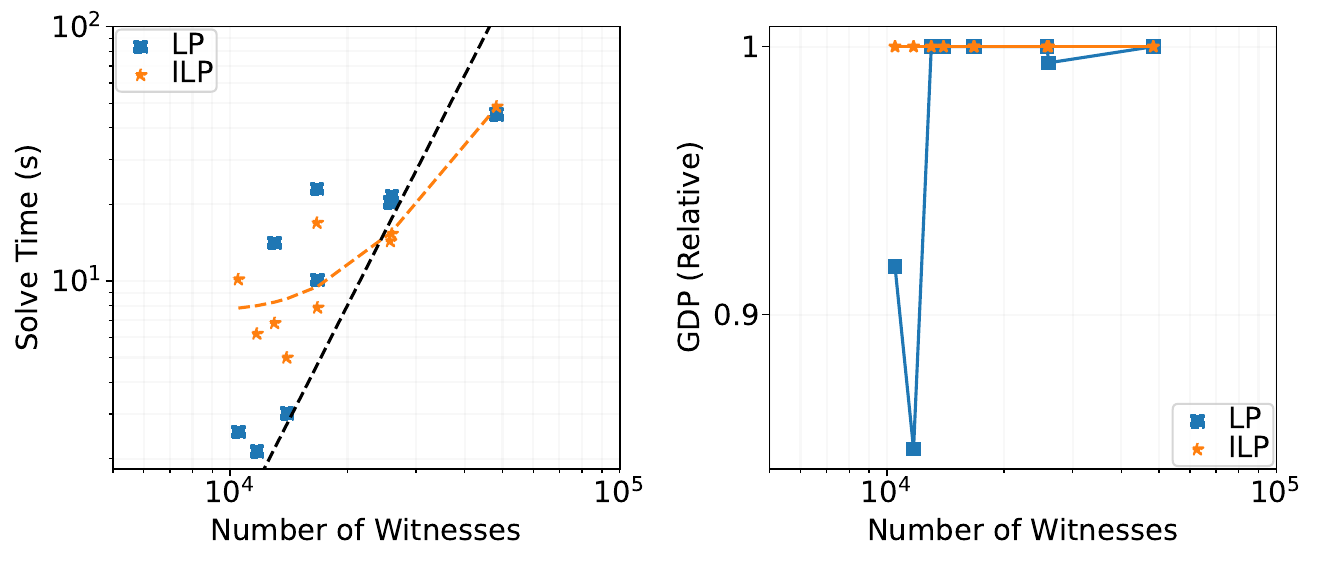}
    \caption{
    (Q4): Performance Evaluation of the Generalized Deletion Propagation over a real-world dataset shows fast solve times (within a minute) and comparable to linear-time (black dashed line) scalability.
    }
    \label{fig:expt-4}
\end{figure}

\section{A Note on System Implementation}
\label{sec:systemimplementation}

Our current proof of concept (see \cref{sec:expts}) uses Python to compute the provenance of query results, and to translate this into an ILP formulation.
This part could be more tightly integrated with existing database systems by leveraging
existing efforts in our community that have been investigating how to repurpose provenance functionality during query execution into PostgreSQL, 
such as Perm~\cite{Glavic:Perm:2010,DBLP:conf/sigmod/GlavicA09},
GProM~\cite{DBLP:journals/debu/ArabFGLNZ17}, or
ProvSQL~\cite{DBLP:journals/pvldb/SenellartJMR18}. 
Furthermore, the extensibility features of today's database systems could be used to add such an ILP solver to the database systems, just user defined functions can be written in programming languages other than the native SQL~\cite{brucato2019scalable,DBLP:journals/cacm/Chaudhuri19,DBLP:journals/pvldb/MaiWABHM24}.
Such integrations of highly sophisticated solvers for solving important database problems have been previously proposed for consistent query answering~\cite{DBLP:conf/icde/DixitK22}
and query optimization~\cite{DBLP:conf/sigmod/Trummer017}, and we believe will become more common.

We believe our unified approach provides an easier integration into existing database infrastructure than prior solutions for several reasons:
\circledwhite{1}~Prior solutions to individual problems use different solution approaches
(e.g., ADP-SS uses dynamic programming,
whereas DP-SS uses reduction to flow).
Thus an integration of all prior work would require \emph{several adaptations}, one for each method.
\circledwhite{2}~For an approach to be natively supported by a relational database 
(thus without user defined functions written in a programming language),
the approach would have to be first-order rewritable.
Among the prior solutions, the only cases we know of that are first-order rewritable are the few PTIME cases for DP-VS~\cite{KimelfeldVW12}.
All other approaches require writing functionality in programming languages other than SQL, just as ours.
\circledwhite{3}~All prior exact methods (except~\cite{makhija2024unified}) are incomplete in that they work only for those conjunctive queries which can be solved in guaranteed PTIME. Only some prior work propose approximation algorithms for the hard cases
(such as the one for DP-VS), yet implementing those require yet other methods.
\circledwhite{4}~In addition, to our approach being the only one that is complete, 
it also has a desirable anytime property: ILP solvers can produce solutions of increasing quality as optimization progresses and are able to provide bounds for how far the current solution is from the optimum.

\section{Future Work}
In this paper, we have introduced a new framework for deletion propagation problems that it is practical, efficient, and complete.
We leave with an interesting conjecture: \emph{All queries for which deletion propagation problems are tractable can be solved in PTIME via our approach}. 
We show that \emph{this conjecture is true for all currently known tractable cases}, and also show an example of a query with self-joins and unions that was not previously studied and that is tractable for our approach, even under bag semantics.
If proven correct, this would be another practically appealing reason for using our approach now: we would get guaranteed PTIME performance already today, even if the proof happens in the future.
We leave open the question of if this conjecture is true for all tractable cases, and the question of the complexity dichotomies of various deletion propagation problems (including the ones introduced in this paper).

\begin{acks} %
This work was supported in part by the National Science Foundation (NSF) under award number IIS-1762268 and IIS-1956096, 
and conducted in part while the authors were visiting the Simons Institute for the Theory of Computing.
\end{acks}

\bibliographystyle{ACM-Reference-Format}
\bibliography{BIB/delprop.bib}

\appendix
\crefalias{section}{appendix}

\section{Notation}
\label{sec:nomenclature}

\begin{table}[h]
\centering
\small
\begin{tabularx}{\linewidth}{@{\hspace{0pt}} >{$}l<{$} @{\hspace{2mm}}X@{}} %
\hline
\textrm{Symbol}		& Definition 	\\
\hline
    \hline
	Q			& Conjunctive query (CQ)	\\
	\D			& Database Instance, i.e.\ a set of tables	\\
    R, S, T    & Relations \\
    x, y, z		& Query variables \\
    \var(R)     & Variables in relation $R$ in a query $Q$ \\
    Q(D)        & A view representing the evaluation of query $Q$ on database $D$\\
    \vec w  	& Witness \\

    \mathcal{Q} & An ordered set of queries \\
    Q^i     & The $i$-th query in $\mathcal{Q}$ \\
    |Q(\D)|  & The number of tuples in the view $Q(\D)$ \\
    \Gamma        & A set of tuples (usually denoted a set of tuples to be deleted from the database) \\
    |\Delta Q(\D, \Gamma)| & $|Q(\D)| - |Q(\D -\Gamma)|$\\
    X[v]        & A binary variable in an (Integer) Linear Program corresponding to a variable $v$\\
    \res & Resilience \\
    \dpss & Deletion Propagation with Source Side Effects \\
    \dpvs & Deletion Propagation with View Side Effects \\
    \adp & Aggregated Deletion Propagation \\
    \swp & Smallest Witness Problem \\
    \gdp & Generalized Deletion Propagation \\
    \gdpilpbasic & A naive ILP for $\gdp$ \\
    \gdpilpmiddle & An ILP for $\gdp$ with wildcard semantics \\
    \gdpilpsmooth / \gdpilp & ILP for $\gdp$ with smoothening constraints - an ILP with tractability guarantees \\
\hline
\end{tabularx}
\caption{Notation table}
\label{tbl:nomenclature}
\end{table}

\section{Proofs for Section~\cref{SEC:ILP}: ILP Framework for GDP}

\thmwildcardilp*

\begin{proof}[Proof \cref{th:ILP:wildcard}]
    We argue that every solution allowed by the wildcard semantics corresponds to a feasible solution of the ILP under naive semantics and vice versa.
    Additionally, we need to show that the optimal value of the ILP under wildcard semantics is the same as the optimal value of the ILP under naive semantics.
    That every solution of the ILP under naive semantics corresponds to a solution of the ILP under wildcard semantics is easy to see since the wildcard semantics is a generalization of the naive semantics.
    To argue the other direction, we need to show that every solution of the ILP under wildcard semantics corresponds to a solution of the ILP under naive semantics with the same optimal value.
    Such a corresponding solution can be obtained by keeping the values of all $X[t]$ variables the same, and replacing all $X[v]$ and $X[w]$ variables to what is implied by the variables under naive semantics.
    We can see such a solution will still be feasible under the naive semantics since the wildcard semantics enforce one-sided guarantees on the values of the variables in the direction of the user constraints for witnesses and view tuples in $\viewdel$ and $\viewpres$.
    We now need to argue about the optimal solutions under naive and wildcard semantics.
    First, we observe that using wildcard semantics can never lead to a solution where the corresponding naive solution has a less optimal objective value. 
    This is due to the fact that we have one-sided guarantees on the values of witnesses and view tuples in $\viewdel$ and $\viewpres$, that ensure that the objective value under wildcard semantics is never strictly better than the objective value under naive semantics.
    Combined with the fact that every solution under naive semantics is a solution under wildcard semantics, we can conclude that the optimal value of the ILP under wildcard semantics is the same as the optimal value of the ILP under naive semantics.
\end{proof}

\thmsmoothenedilp*

\begin{proof}[Proof \cref{th:ILP:smoothened}]
    Under wildcard semantics, each solution set of interventions in the database corresponds to multiple ILP solutions, since the value of $X[w]$ and $X[v]$ variables are not necessarily faithful to the variables.
    For $\querypres$, we can assume that a solution with wildcard semantics exists such that for every view tuple $v \in \querypres(\D)$, there is only one witness such that $v \subseteq w$ and $X[w]=0$ i.e. we can assume that for each $v \in \querypres(\D)$, there is a unique $w$ that is marked as preserved (note that more may be preserved, but unmarked due to the wildcard semantics).
    The optimal value of the ILP under this assumption is the same as the optimal value of the ILP under wildcard semantics (and hence naive semantics).
    Due to this assumption, we can say that a given tuple and view tuple, only one witness is preserved i.e. $\sum_{i \in [1,k]} X[w_i] \geq k-1$ where $w_1 \hdots w_k$ are the witnesses that correspond to $v$ and contain $t$. Now we can also enforce that $X[t] \leq (\sum_{i \in [1,k]} X[w_i]) - (k-1)$, since all but $1$ values of $X[w]$ are set to $1$, and only the final value decides the upper bound on $X[t]$.
    This is equivalent to adding the smoothing constraint, and hence addition of the smoothing constraint does not change the optimal solution of the ILP.
\end{proof}

\section{Proofs for ~\cref{SEC:TRACTABILITY}: Recovering Existing Tractability Results}

\dpssptime*

\begin{proof}[Proof \cref{prop:dpssptime}]
    Prior work has shown the construction of an ILP $\resilp$ that has the property that $\reslpparam{Q, D} = \res(Q, D)$ for all $Q$ for which $\res(Q,D)$ is known to be in PTIME - both under set and bag semantics.
    We simply reuse this result, and show that $\gdpdpsslp$ is identical to $\reslp$ for the same query $Q$ and database $D$.
    We know that in the formulation of $\gdpdpsslp$, both $\querymax$ and $\querypres$ are empty. 
    We also know that all the output tuples must be deleted, thus we can replace all variables $X[v]$ for $v \in \querydel(\D)$ with the constant $1$ without changing the optimal solution of the ILP. 
    Since each output tuple in $\querymin(\D)$ has a 1-to-1 correspondence with the input tuples in $D$, the minimization objective directly corresponds to the number of input tuples that are deleted.
    Thus, in this case $\gdpilp$ simplifies to look like:
    \begin{align*}
        \min &\sum_{t \in \D} X[t] \\
        \textsf{subject to}& \\
        \sum_{t \in w}& X[t] \geq 1 {\qquad\forall w \in Q^F(\D)}\\
    \end{align*}
    We observe that this ILP is identical to the ILP $\resilp$ that has been proposed for the resilience problem~\cite{makhija2024unified}, and thus the LP relaxation of $\gdpdpssilp$ is integral for all queries $Q$ for which $\dpss$ is known to be in PTIME~\cite[Theorem 8.6 and 8.7.]{makhija2024unified}.
\end{proof}

\dpvsptime*

\begin{proof}[Proof \cref{prop:dpvsptime}]
    $\dpvs$ can be solved in PTIME for a query $Q$ over an arbitrary database $\D$ if and only if it has the head domination property~\cite{KimelfeldVW12}.
    It is also known that for queries that have the head domination property, the optimal solution for $\dpvs$ is side effect free \cite[Proposition 3.2.]{KimelfeldVW12}.
    Thus, in our formulation $\gdpilp$ always has an optimal solution of $1$ - since the only output tuple that is deleted is the one specified by the user.
    Since we know that $\gdpdpvslp$ is always a lower bound for the true optimal solution, it suffices for us to show that the LP relaxation of $\gdpdpvslp$ is not $< 1$.
    We know that the objective function is a sum of $X[v]$ variable where each $v$ is an output tuple in $\querymax(\D)$ and $X[v]$ takes on value $0$ or $1$ and thus can never be negative.    
    We also know that for a $v$ that the user would like to delete via $\dpvs$, $X[v]$ must be set to $1$.
    Thus, the solution to $\gdpdpvslp$ is never $<1$ and never more than the optimal solution (which is always $1$), and thus the LP relaxation is integral and equal to the optimal solution.
\end{proof}

\swpptime*

\begin{proof}[Proof \cref{prop:swpptime}]
    In order to construct $\swp$ as a $\gdpilp$, we set $\viewdel = \viewmin = \emptyset$ since there are no views from which output must be deleted, or deletions must be minimized.
    We set $\viewpres = \{Q(\D)\}$, and $\viewmax = [R \in D]$. $\kpres$ is set to $|Q(\D)|$ as we want to preserve all tuples in $Q(\D)$, and $\kdel$ is simply $0$.

    We can then construct the ILP $\gdpswpilp$ as follows:
    \begin{align*}
        \min & -\sum_{v \in \viewmax} X[v]\\
        \textsf{subject to}& \\
        &\sum_{v \in Q(\D)} X[v] \leq |Q(\D)| - \kpres\\
        & X[t] \leq X[w], t \in w \\ 
        & X[w] \leq \sum_{t \in w} X[t] \\
        & X[v] \leq X[w], v \subseteq w \\
        & 1+\sum_{v \subseteq w} \big(X[w] - 1\big) \leq X[v] \\
        & X[t], X[w], X[v] \in \{0,1\} \quad \forall t,w,v
    \end{align*}

    We know that each output tuple in $\viewmax$ has a 1-to-1 correspondence with the input tuples in $D$ (since $\viewmax$ is simply a union of all input relations).
    Thus, we can replace all view variables $X[v]$ where $v \in \viewmax$ with simply the corresponding input tuple variables $X[t]$.
    Now the function of the propagation constraints is simply to propagate the deletions of output tuples in $\viewpres$ to the input tuples in $D$. 
    Since the user constraint provides an upper bound for $X[v]$ for $v \in \viewpres$, the propagation constraints must only provide a lower bound for $X[v]$ via $X[w]$ and then a lower bound for $X[w]$ via $X[t]$.

    Thus, we can greatly simplify the ILP to:
    \begin{align*}
        \min & -\sum_{t \in \D} X[t]\\
        \textsf{subject to}& \\
        &\sum_{v \in Q(\D)} X[v] \leq |Q(\D)| - |Q(\D)|\\
        & 1+\sum_{v \subseteq w} \big(X[w] - 1\big) \leq X[v] \\
        & X[t] \leq X[w], t \in w \\ 
        & X[t], X[w], X[v] \in \{0,1\} \quad \forall t,w,v
    \end{align*}

    Since $\swp$ primarily deals with preservations rather than deletions, we can introduce a variable $Y[v]$ that captures if a tuple $v$ is preserved, i.e., it is set to 1 if the tuple is preserved, and 0 otherwise. We can see that $Y[v] = 1- X[v]$ and build a corresponding ILP with $Y$ variables instead of $X$ variables.
    This resulting ILP is identical to the one with $X$ variables in that it will have the same optimal solution, and any solution of an ILP or LP relaxation from one can be mapped to another by simply applying the equality $X[v] = 1 - Y[v]$.

    \begin{align*}
        \min & -\sum_{t \in \D} 1 - Y[t]\\
        \textsf{subject to}& \\
        &\sum_{v \in Q(\D)} 1 - Y[v] \leq |Q(\D)| - |Q(\D)|\\
        & 1 - Y[t] \leq 1 - Y[w], t \in w \\ 
        & 1+\sum_{v \subseteq w} \big(1 - Y[w] - 1\big) \leq 1 - Y[v] \\
        & Y[t], Y[w], Y[v] \in \{0,1\} \quad \forall t,w,v
    \end{align*}

    Simplifying this ILP, we get:

    \begin{align*}
        \min & -|D| + \sum_{t \in \D} Y[t]\\
        \textsf{subject to}& \\
        & |Q(D)| \leq \sum_{v \in Q(\D)} Y[v] \\
        & \sum_{v \subseteq w} \big(- Y[w] \big) \leq - Y[v] \\
        & Y[w] \leq Y[t], t \in w \\ 
        & Y[t], Y[w], Y[v] \in \{0,1\} \quad \forall t,w,v
    \end{align*}

    We can see through the constraint $|Q(D)| \leq \sum_{v \in Q(\D)} Y[v]$ that every $Y[v]$ value must necessarily be set to $1$ to satisfy the constraint, whether in the ILP or in the LP relaxation. 
    We can thus simplify further to:

    \begin{align*}
        \min & \sum_{t \in \D} Y[t]\\
        \textsf{subject to}& \\
        & 1 \leq \sum_{v \subseteq w} \big(Y[w]\big) \\
        & Y[w] \leq Y[t], t \in w \\ 
        & Y[t], Y[w], Y[v] \in \{0,1\} \quad \forall t,w,v
    \end{align*}

    Another way to see this simplified linear program is that it constrains that at least one witness must be preserved for every output tuple, and that if a witness is preserved, then all tuples in it must be preserved.

    We now show that for queries that have the head cluster property, the solution of $\gdpswplp$ is always integral. 
    We first restate the definition of the head cluster property.

    \begin{definition}[Existential Connectivity Graph $G^\exists_Q$]
        The existential-connectivity graph of a query $Q$ is a graph $G^\exists_Q$ where each relation $R_i \in Q$ with $\attr(R_i)$ - $\head(Q) \neq \emptyset$ is a vertex, and there is an edge between $R_i$ and $R_j$ if $\attr(R_i) \cap \attr(R_j) - \head(Q) \neq \emptyset$.
    \end{definition}
    We can find the connected components of $G^\exists_Q$ by applying search algorithm on $G^\exists_Q$, and finding all connected components for $G^\exists_Q$. 
    Let $E_1, E_2, \hdots, E_k \subseteq \rels(Q)$ be the connected components of $G^\exists_Q$, each corresponding to a subset of relations in Q.

    \begin{definition}[Head Cluster Property]
        A query $Q$ has the head cluster property if for every pair of relations $R_i, R_j \in Q$ with $\head(R_i) \neq \head(R_j)$, it must be that $R_i$ and $R_j$ are in different connected components of $G^\exists_Q$.
    \end{definition}

    In other words, the head cluster property ensures that all relations in the same connected component of $G^\exists_Q$ have exactly the same head variables. 

    First let's assume we have a query with a single connected component in $G^\exists_Q$. 
    The head variables in this case are the same for all relations in the query, and are the same as the head variables of the query.
    Due to this, each input tuple contributes to a single output tuple, and thus we can treat each projection to be preserved independently.
    In other words, we can reduce our problem to preserving a single projection. Consider that the query has $m$ relations. 
    In a self-join free query, at least one tuple from each relation must be preserved, and a solution can be obtained that preserves exactly $m$ tuples by preserving any one witness arbitrarily. 
    Thus, the optimal solution to the ILP for a preserving a single projection independently is always to preserve exactly $m$ tuples. 
    We need to prove that the LP relaxation of $\gdpswplp$ never has an optimal value smaller than $m$.
    For every relation $R_i$ in the query, we show that the sum of $X[t]$ variables for $t\in R_i$ is always at least $1$.
    This is because for each witness we associate a corresponding $R_i$ input tuple that has at least the fractional value assigned to $X[w]$ (if a tuple of $R_i$ corresponds to multiple $w$ then it takes on sum of the fractional values of $X[w]$). Since the sum of $X[w]$ variables is at least $1$, the sum of $X[t]$ variables for $t\in R_i$ is at least $1$. 
    Repeating this argument for all relations in the query, we can see that the LP relaxation of $\gdpswplp$ at least $m$, and thus never better than the optimal solution of the ILP.

    Now let's consider a query with multiple connected components in $G^\exists_Q$. 
    We argue that we can treat each connected component independently, and the optimal solution of the ILP is the sum of the optimal solutions of the ILP for each connected component.
    Due to the fact that we deal with self-join free queries, no relations participate in multiple connected components, and thus we can partition the input to the ILP into disjoint sets of input relations. 

    We can show that the optimal value of the LP relaxation of $\gdpswplp$ is equal to the sum of the optimal values of $\gdpswplp$ for each connected component (which we have shown above to be integral).
    First we see naturally that the LP relaxation of $\gdpswplp$ naturally decomposes into the sum of the LP relaxations of $\gdpswplp$ for each connected component - since if the overall output is preserved, each connected component must be preserved.
    In the other direction, we want to show that if we add up the solutions for each connected component, we get a feasible solution for the LP relaxation of $\gdpswplp$. 
    This follows from the fact that relations in different connected components share only head variables and thus the preserved tuples in one connected component always join with the preserved tuples in another connected component since they share the same head variables. 
    Thus the sum of the solutions for each connected component is a feasible solution for the LP relaxation of $\gdpswplp$, and thus the LP relaxation is integral and equal to the optimal solution of the ILP.
\end{proof}

\adpptime*

\begin{proof}[Proof \cref{prop:adpptime}]
    The proof of tractability of $\adp$~\cite{ADP} is divided into two base cases and two types of recursive decompositions. 
    The proof of the integrality of the LP relaxation of $\gdpadplp$ follows the same structure as the original proof of tractability of $\adp$ ~\cite{ADP}. 

    The two bases cases are for (1) Boolean queries and (2) queries with a \emph{singleton} relation i.e. a relation whose variables are either a subset of the relations of all other relations in the query, or a subset or superset of the head variables of the query. For base case (1), we see that this reduces exactly to the resilience problem, and we can simply use \cref{prop:dpssptime} to show the integrality of $\gdpadpilp$.
    For base case (2), we observe that there is always an optimal solution of $\adp$ that deletes tuples only from the singleton relation. In the language of the resilience problem \cite{DBLP:journals/pvldb/FreireGIM15}, the singleton relation dominates all other relations in the query. 
    The tuples of the singleton relations can then be removed from the optimization problem.
    The resulting problem now contains only one tuple per witness.
    The constraint matrix of this ILP never contains a cycle - since no two input tuples share a witness, and thus do not form a cycle in the constraint matrix. 
    The constraint matrix is thus balanced\cite{schrijver1998theory}, and through well known results in optimization theory, is known to have an integral solution for any integral objective function. 

    The two types of recursive decompositions are for (1) queries with a universal attribute i.e., an attribute that appears in all relations of the query, and (2) disconnected queries. 
    We apply that this decomposition in a similar manner to the proof of \cref{prop:swpptime}, and show that an optimal solution of $\gdpadplp$ can be obtained by treating each decomposition of the query independently.
    Thus, the overall LP relaxation is always equal to the sum of the LP relaxations of the decomposed queries, and by applying these decompositions recursively and using the integral base cases, we can show that the LP relaxation of $\gdpadplp$ is integral for all queries for which $\adp$ is known to be in PTIME.

    However, a difference in the decomposition performed here vs the proof of \cref{prop:swpptime} is that the value of $k$ in the decomposition of $\adp$ must also be split across the decomposed queries. 
    We do not actually need to run the decomposed ILP over all possible splits of $k$, however we simply prove that for any possible split of $k$, the solutions of the decomposed ILPs leads to a feasible solution of the original ILP that is (recursively) known to be integral.
    The optimal split of $k$ across the decomposed queries (whatever it may be) is also then always recovered by the original LP relaxation and hence the LP relaxation is integral and equal to the optimal solution of the ILP for all queries for which $\adp$ is known to be in PTIME (for all queries for which applying these decompositions leads to one of the two base cases).
\end{proof}

\section{New Tractability Results}
\label{SEC:TRACTABILITYNEW}
\label{sec:beyond-CQs}

In this section, we show an example of a query that contains self-joins and unions, and is tractable under bag semantics for the $\dpvs$ and $\swp$ problems.
We also show that this tractability can be recovered by posing the problems in the $\gdp$ framework.
Bag semantics and queries with self-joins and unions are known to be challenging to analyze for deletion propagation problems, and papers have been written with the sole purpose of making progress on understanding the tractability landscape in this complicated settings \cite{DBLP:conf/pods/FreireGIM20,DBLP:conf/lics/BodirskySL24}. 
This section (and the one query presented in it) are meant to act as a proof of concept that the $\gdpilp$ framework can be an invaluable tool to help understand and recover tractability results for various DP problems in these complicated settings. 

We show that $\dpss$, $\dpvs$, $\swp$ and $\adp$ are tractable for $\qsjtriangletwochain$ \Cref{eqn:qsj} under bag semantics.
We know from prior work that this query is hard for $\res$\cite{DBLP:conf/pods/FreireGIM20}\footnote{It reduces to the boolean SJ-chain query $Q()\datarule R(x,y),R(y,z)$, which is shown to be hard}, and hence can infer that it is hard for $\dpss$ and $\adp$ as well - since they are both generalizations of $\res$.
\begin{equation}\qsjtriangletwochain(a) \datarule R(x, a, b),
R(x, b, c), R(x, c, a) \cup R(x, e, f), R(x, f, g)\label{eqn:qsj}\end{equation}

\begin{restatable}{proposition}{sjqueryptime}
    For all database instances $\D$ under bag semantics: 
    (1) $\gdpdpvslpparam{\qsjtriangletwochain, \D} = \dpvs(\qsjtriangletwochain, \D, t)$ for an arbitrary $t$ in the view 
    (2) $\gdpswplpparam{\qsjtriangletwochain, \D} = \swp(\qsjtriangletwochain, \D)$ 
    \label{prop:sjqueryptime}
\end{restatable}

\begin{proof}[Proof \cref{prop:sjqueryptime}]
    We look at each problem in turn.

    For (1) DP-VS, we simply show that like in \cref{prop:dpvsptime}, the optimal solution of $\gdpdpvsilp$ is always $1$ and the LP relaxation thus cannot take a lower value. 
    In other words, a solution can be obtained that is side effect free. 
    If we simply delete all facts that contribute to the output tuple, we will obtain a side effect free solution since every relation contains a super set of the head variables of $\qsjtriangletwochain$.
    Thus, every tuple that is consistent with the output tuple to be deleted, simply cannot be consistent with or contribute to any other output tuple, and thus this solution is side effect free.
    The rest of the proof is identical to \cref{prop:dpvsptime}.

    For (2) SWP, we show that like in \cref{prop:swpptime}, each input tuple contributes to a single output tuple and thus each projection can be treated independently. 
    We also observe that the first sub-query of $\qsjtriangletwochain$ i.e. $\qsjtriangletwochain_1(x) \datarule R(x, a, b), R(x, b, c), R(x, c, a)$ \emph{is dominated by} the second sub-query i.e., $\qsjtriangletwochain_2(x) \datarule R(x, e, f), R(x, f, g)$. 
    This means that preserving all output tuples of $\qsjtriangletwochain_2$ will also preserve all output tuples of $\qsjtriangletwochain_1$.
    In other words, the constraints enforced to preserve the output tuples of $\qsjtriangletwochain_2$ automatically imply the constraints that are enforced to preserve the output tuples of $\qsjtriangletwochain_1$, and it suffices to reason about the preservation of the output tuples of $\qsjtriangletwochain_2$.

    We then also notice that each input tuple contributes to a single output tuple, and thus we can treat each projection independently (just like in \cref{prop:swpptime}).
    The difference here is that due to the self-joins in the query, a single input tuple can contribute to a witness ``multiple times''. 
    However, this does not make a difference in the phrasing of Propagation Constraint 4, since the constraint looks at each input tuple independently. 
    It may be a witness has fewer tuples than atoms in the query (since tuples may join with themselves), however, in terms of the ILP this just means that there are fewer constraints to enforce, and the ILP is still integral.
    Due to this fact, the remainder of the proof is identical to \cref{prop:swpptime}, and we can show that the LP relaxation of $\gdpswplp$ is integral for $\qsjtriangletwochain$.
\end{proof}

\section{Additional Experiments}

In addition to the four questions studied in \cref{sec:expts}, we look at two additional questions:
(Q5) What is the scalability of $\gdpilp$ for PTIME cases with respect to different input parameters such as the number of tuples, the number of relations, and the maximum domain size?
(Q6) What is the memory usage of $\gdpilp$ for various PTIME cases of DP problems?

\smallsection{(Q5) Scalability of $\gdpilp$ for changing parameters}
\Cref{fig:scalabilitydomain}
show the scalability of $\gdpilp$ for $\dpvs,\swp$ and $\adp$ over different maximum domain sizes 
($10^2, 10^3, 10^4$, $10^5$),
and
\cref{fig:scalabilityjoins}
for $k$-star queries with increasing number of joins 
($\qthreeray$, $\qfourray$, $\qfiveray$, $\qsixray$).
We see that the scalability in both cases is very well predicted by the number of witnesses, irrespecitve of the domain size or number of tuples.
Increasing the maximum domain size leads to a less dense instance that has fewer witnesses and is thus easier to solve than a more dense instance with the same number of tuples. 
Similarly, for changing the number of joins in the query.

\smallsection{(Q6) Memory consumption of $\gdpilp$}
We conducted an additional experiment \cref{fig:memory}
that shows the measured space consumption of $\gdpilp$
using the psutil library\footnote{\url{https://pypi.org/project/psutil/}}.
\Cref{fig:memory} shows the memory consumption of $\gdpilp$ for a $\dpvs, \swp$ and $\adp$ problems
as a function of the number of witnesses, parameterized for various maximum domain sizes and queries with increasing number of joins. 
We see that although the memory consumption is not guaranteed to be sublinear, in practice it is sublinear w.r.t.\ the number of witnesses.

\begin{figure}[H]
    \centering
    \begin{subfigure}[b]{0.02\textwidth}
        \begin{turn}{90}\textbf{\qquad\qquad\quad\circled{\footnotesize\normalfont V}$\dpvs$}\end{turn}
    \end{subfigure}
    \begin{subfigure}[b]{0.47\columnwidth}
        \includegraphics[scale=0.33]{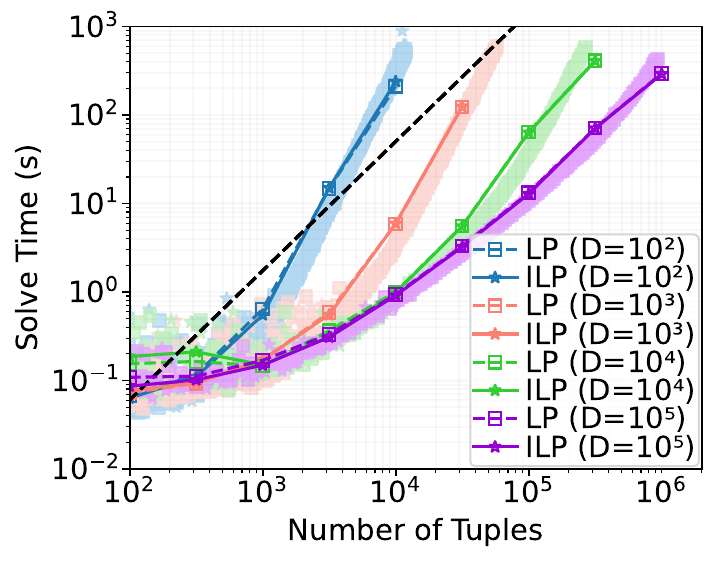}
        \caption{}
    \end{subfigure}
    \begin{subfigure}[b]{0.47\columnwidth}
        \includegraphics[scale=0.33]{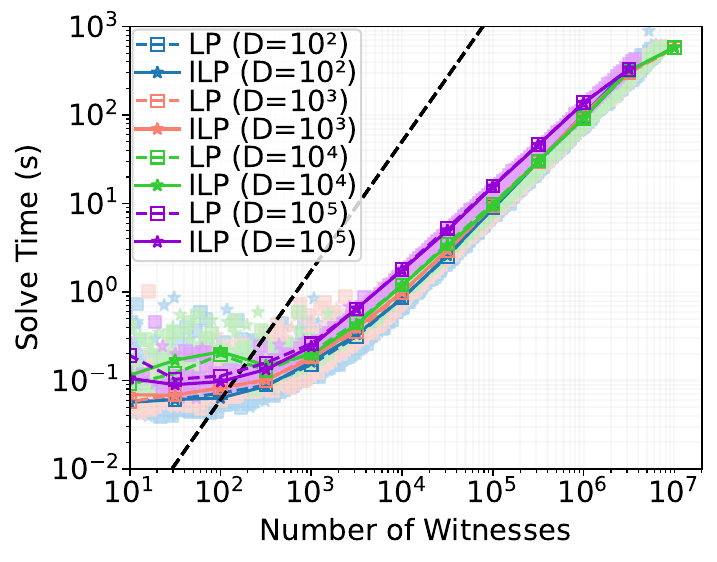}
        \caption{}
    \end{subfigure}
    \\
    \begin{subfigure}[b]{0.02\textwidth}
        \begin{turn}{90}\textbf{$\qquad\qquad\quad\circled{\footnotesize\normalfont S}\swp$}\end{turn}
    \end{subfigure}
    \begin{subfigure}[b]{0.47\columnwidth}
        \includegraphics[scale=0.33]{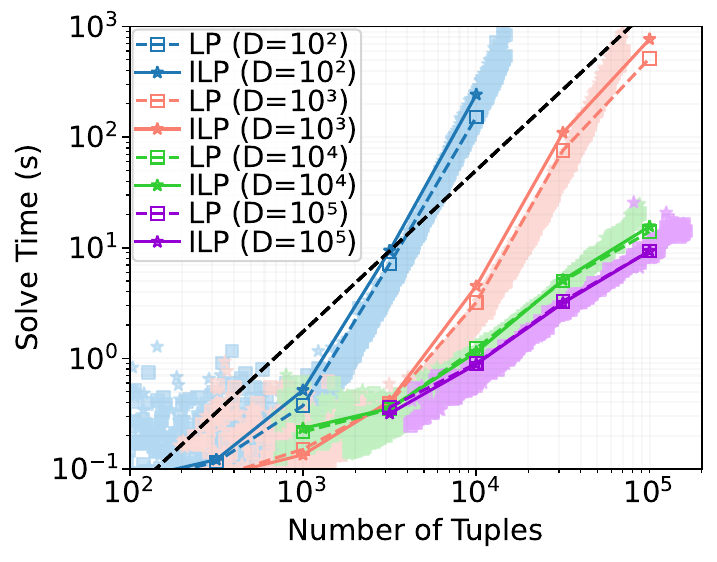}
        \caption{}
    \end{subfigure}
    \begin{subfigure}[b]{0.47\columnwidth}
        \includegraphics[scale=0.33]{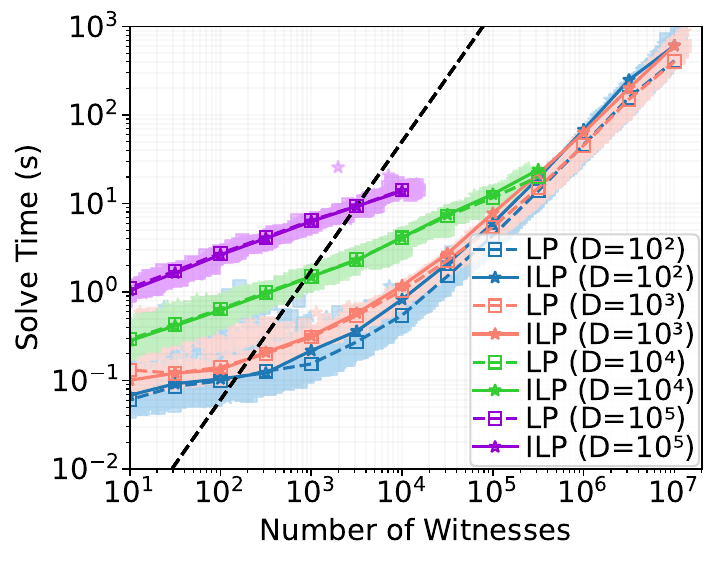}
        \caption{}
    \end{subfigure}
    \begin{subfigure}[b]{0.02\textwidth}
        \begin{turn}{90}\textbf{$\qquad\qquad\quad\circled{\footnotesize\normalfont A}\adp$}\end{turn}
    \end{subfigure}
    \begin{subfigure}[b]{0.47\columnwidth}
        \includegraphics[scale=0.33]{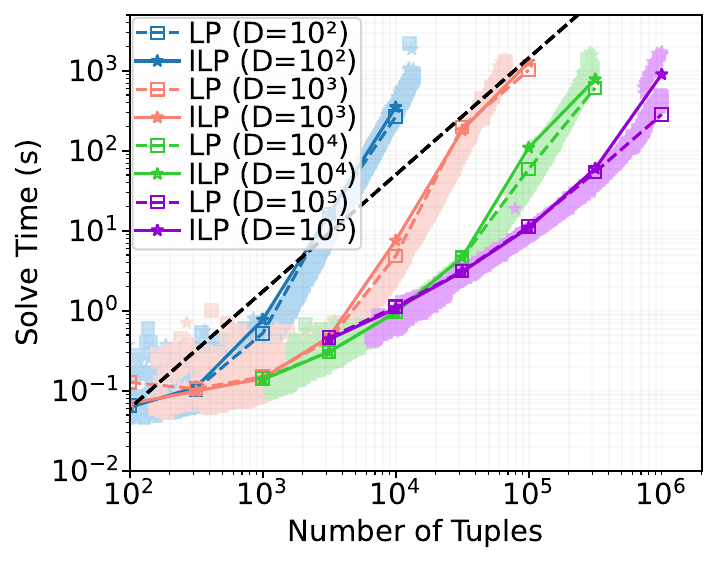}
        \caption{}
    \end{subfigure}
    \begin{subfigure}[b]{0.47\columnwidth}
        \includegraphics[scale=0.33]{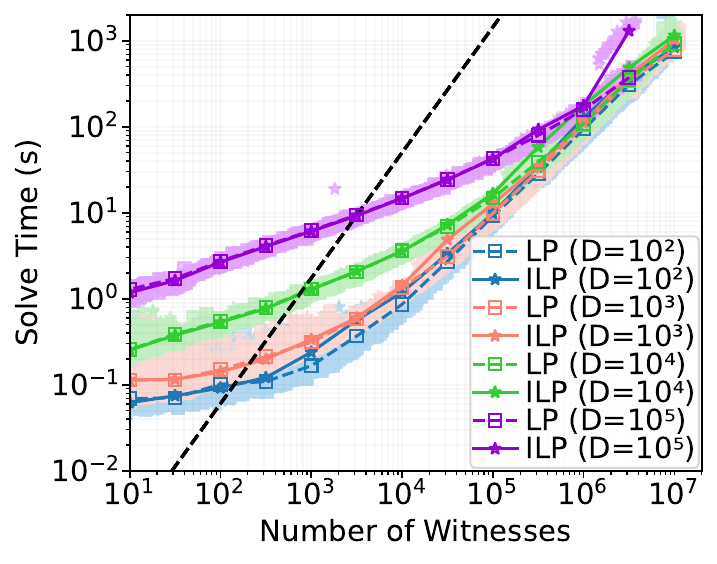}
        \caption{}
    \end{subfigure}
    \caption{
    (Q5a): Scalability of $\gdpilp$ for different domain sizes: we show the scalability with respect to the number of tuples (left) and the number of witnesses (right) for $\dpvs$ (top), $\swp$ (middle), and $\adp$ (bottom). We find that scalability of $\gdpilp$ is very well predicted by the number of witnesses, irrespective of the number of tuples or the domain size used. 
    }
    \label{fig:scalabilitydomain}
\end{figure}

\begin{figure}[H]
    \centering
    \begin{subfigure}[b]{0.02\textwidth}
        \begin{turn}{90}\textbf{\qquad\qquad\quad\circled{\footnotesize\normalfont V}$\dpvs$}\end{turn}
    \end{subfigure}
    \begin{subfigure}[b]{0.47\columnwidth}
        \includegraphics[scale=0.33]{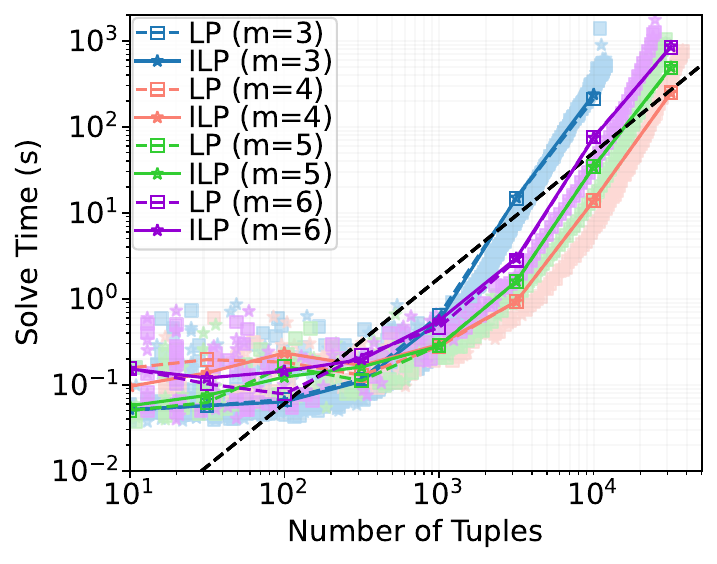}
        \caption{}
    \end{subfigure}
    \begin{subfigure}[b]{0.47\columnwidth}
        \includegraphics[scale=0.33]{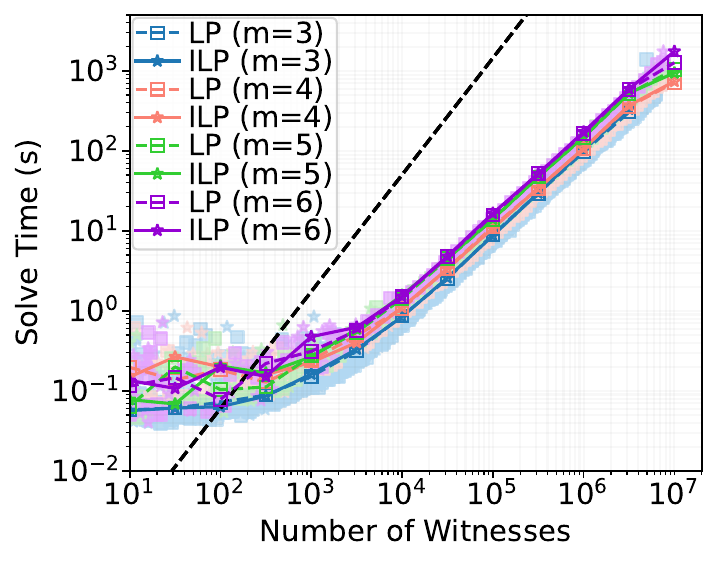}
        \caption{}
    \end{subfigure}
    \\
    \begin{subfigure}[b]{0.02\textwidth}
        \begin{turn}{90}\textbf{$\qquad\qquad\quad\circled{\footnotesize\normalfont S}\swp$}\end{turn}
    \end{subfigure}
    \begin{subfigure}[b]{0.47\columnwidth}
        \includegraphics[scale=0.33]{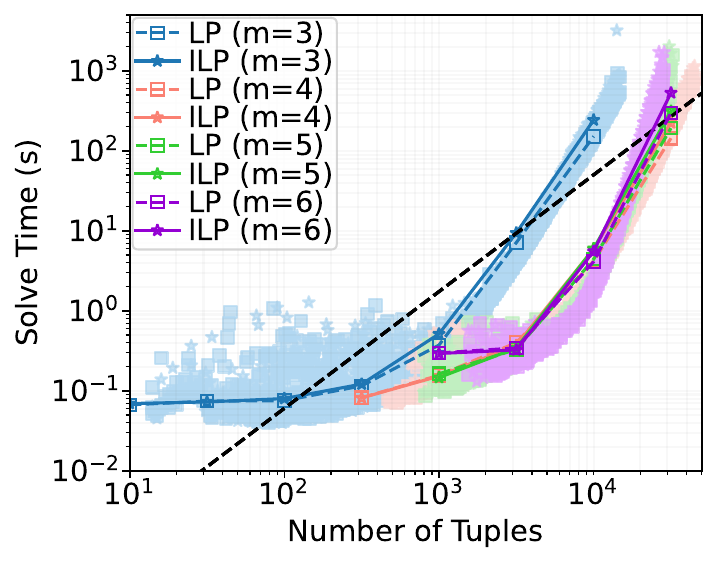}
        \caption{}
    \end{subfigure}
    \begin{subfigure}[b]{0.47\columnwidth}
        \includegraphics[scale=0.33]{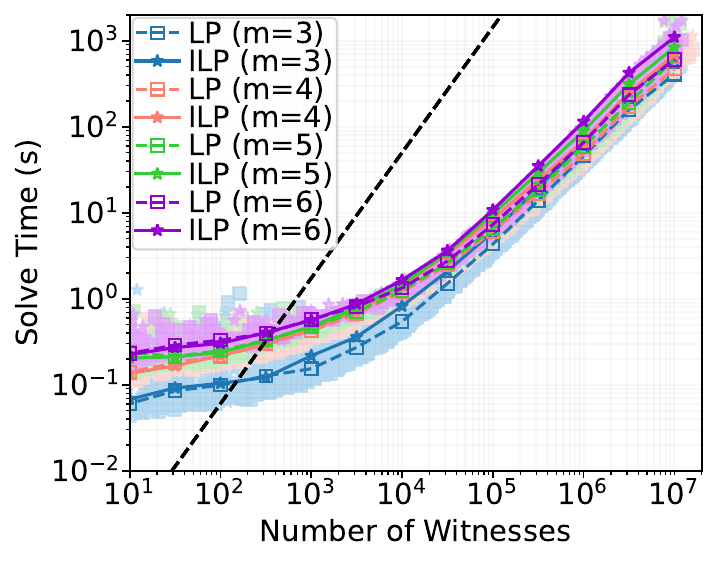}
        \caption{}
    \end{subfigure}
    \begin{subfigure}[b]{0.02\textwidth}
        \begin{turn}{90}\textbf{$\qquad\qquad\quad\circled{\footnotesize\normalfont A}\adp$}\end{turn}
    \end{subfigure}
    \begin{subfigure}[b]{0.47\columnwidth}
        \includegraphics[scale=0.33]{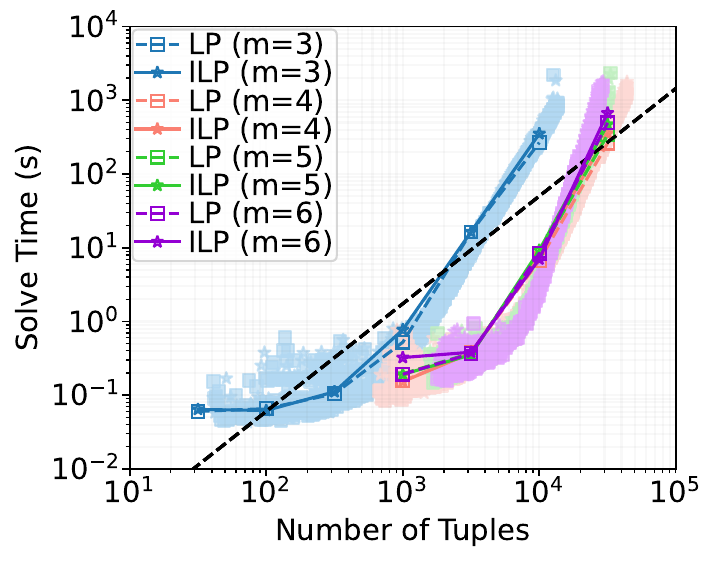}
        \caption{}
    \end{subfigure}
    \begin{subfigure}[b]{0.47\columnwidth}
        \includegraphics[scale=0.33]{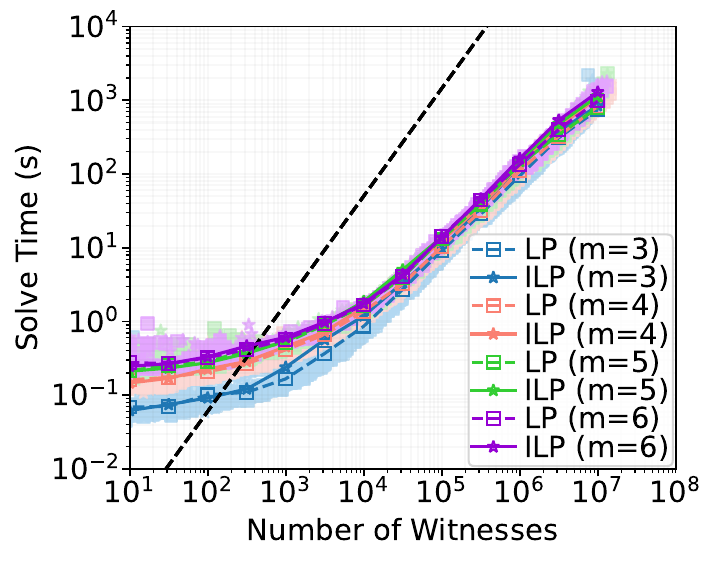}
        \caption{}
    \end{subfigure}
    \caption{
    (Q5b): Scalability of $\gdpilp$ for different query size: we show the scalability with respect to the number of tuples (left) and the number of witnesses (right) for $\dpvs$ (top), $\swp$ (middle), and $\adp$ (bottom). We find that scalability of $\gdpilp$ is very well predicted by the number of witnesses, irrespective of the number of tuples or the number of joins in the query (given by m).
    }
    \label{fig:scalabilityjoins}
\end{figure}

\begin{figure}[H]
    \centering
    \begin{subfigure}[b]{0.02\textwidth}
        \begin{turn}{90}\textbf{\qquad\qquad\quad\circled{\footnotesize\normalfont V}$\dpvs$}\end{turn}
    \end{subfigure}
    \begin{subfigure}[b]{0.47\columnwidth}
        \includegraphics[scale=0.33]{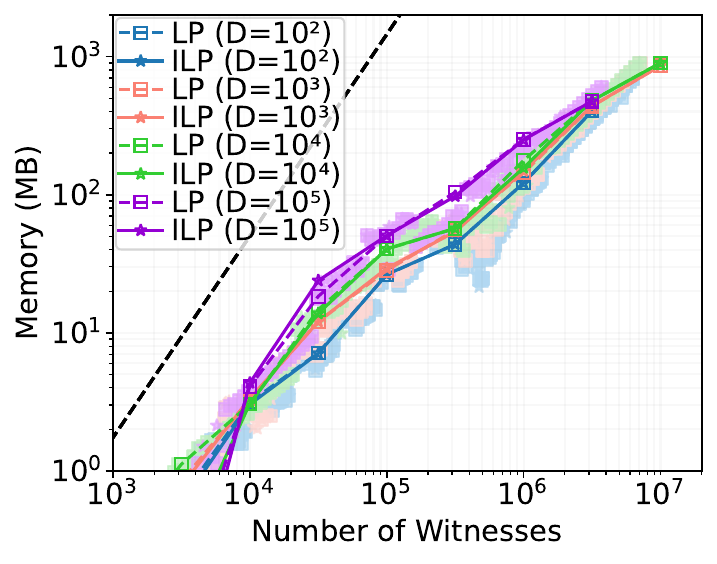}
        \caption{}
    \end{subfigure}
    \begin{subfigure}[b]{0.47\columnwidth}
        \includegraphics[scale=0.33]{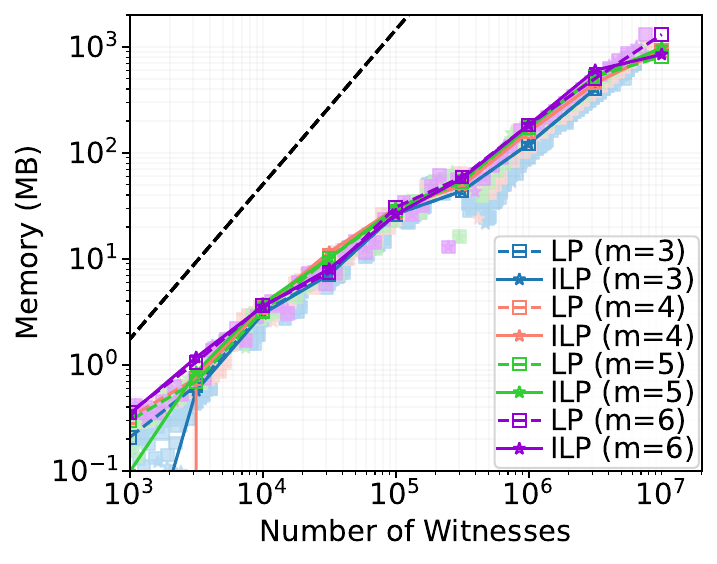}
        \caption{}
    \end{subfigure}
    \\
    \begin{subfigure}[b]{0.02\textwidth}
        \begin{turn}{90}\textbf{$\qquad\qquad\quad\circled{\footnotesize\normalfont S}\swp$}\end{turn}
    \end{subfigure}
    \begin{subfigure}[b]{0.47\columnwidth}
        \includegraphics[scale=0.33]{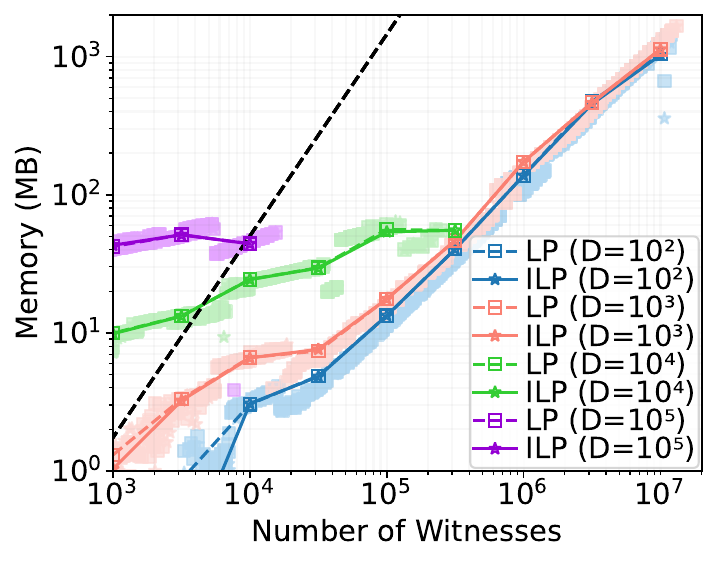}
        \caption{}
    \end{subfigure}
    \begin{subfigure}[b]{0.47\columnwidth}
        \includegraphics[scale=0.33]{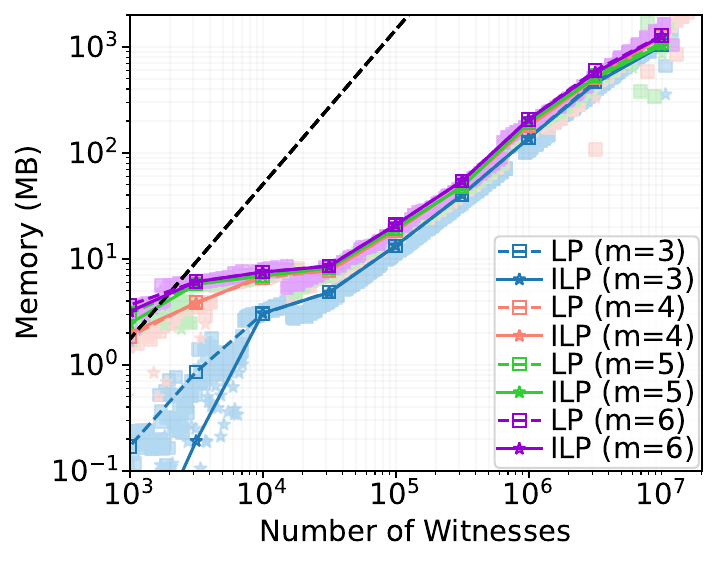}
        \caption{}
    \end{subfigure}
    \begin{subfigure}[b]{0.02\textwidth}
        \begin{turn}{90}\textbf{$\qquad\qquad\quad\circled{\footnotesize\normalfont A}\adp$}\end{turn}
    \end{subfigure}
    \begin{subfigure}[b]{0.47\columnwidth}
        \includegraphics[scale=0.33]{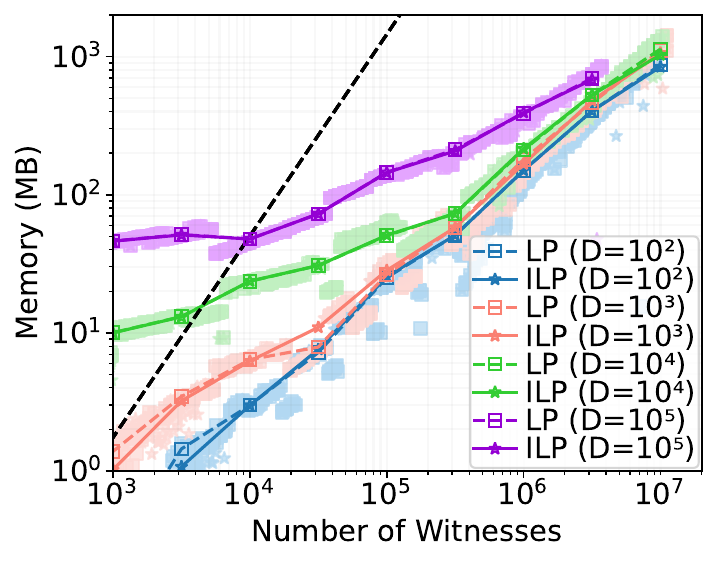}
        \caption{}
    \end{subfigure}
    \begin{subfigure}[b]{0.47\columnwidth}
        \includegraphics[scale=0.33]{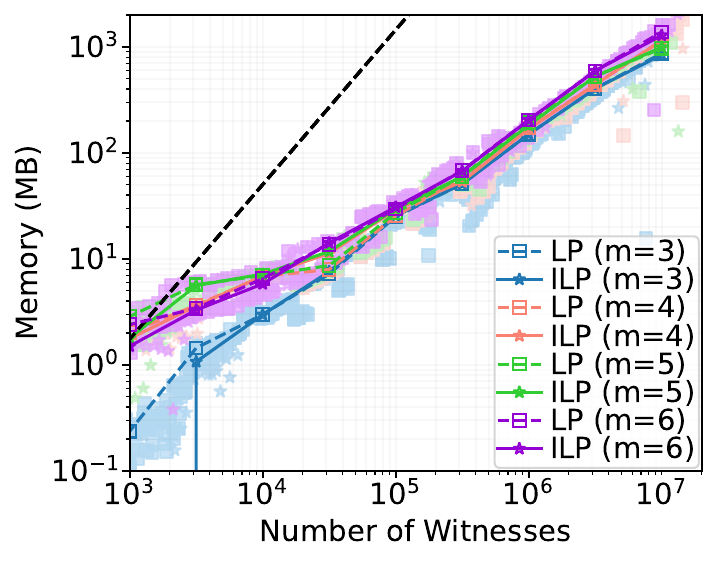}
        \caption{}
    \end{subfigure}
    \caption{
    (Q6): Memory consumption of $\gdpilp$ for three prior studied problems 
    \protect\circled{\footnotesize\normalfont V}~$\dpvs$, 
    \protect\circled{\footnotesize\normalfont S}~$\swp$, and 
    \protect\circled{\footnotesize\normalfont A}~$\adp$. 
    The left column shows the memory consumption over different instances with increasing domain sizes, and the right column shows the memory consumption over different queries with increasing number of joins.
    In all cases, $\gdpilp$ has sublinear memory consumption w.r.t.\ the number of witnesses.
    }
    \label{fig:memory}
\end{figure}           

\end{document}